\documentclass[twocol]{ametsoc}

\usepackage{float}
\usepackage{natbib}
\usepackage{subfigure}
\usepackage{amsthm} 

\newtheorem{theorem}{Theorem}

\newcommand{\EEA}{\end{eqnarray}}
\newcommand{\BEA}{\begin{eqnarray}}
\newcommand{\comment}[1]{}

\journal{mwr}

\bibpunct{(}{)}{;}{a}{}{,}

\title{Correcting biased observation model error in data assimilation}

\authors{Tyrus Berry}
\affiliation{Department of Mathematical Sciences, George Mason University, USA.}

\extraauthor{John Harlim\correspondingauthor{Department of Mathematics, the Pennsylvania State University, 109 McAllister Building, University Park, PA 16802-6400, USA}}
\extraaffil{Department of Mathematics and Department of Meteorology and Atmospheric Science, the Pennsylvania State University, USA}

\email{jharlim@psu.edu}

\abstract{While the formulation of most data assimilation schemes assumes an unbiased observation model error, in real applications, model error with nontrivial biases is unavoidable. A practical example is the error in the radiative transfer model (which is used to assimilate satellite measurements) in the presence of clouds.  As a consequence, many (in fact 99\%) of the cloudy observed measurements are not being used although they may contain useful information. This paper presents a novel nonparametric Bayesian scheme which is able to learn the observation model error distribution and correct the bias in incoming observations.  This scheme can be used in tandem with any data assimilation forecasting system. The proposed model error estimator uses nonparametric likelihood functions constructed with data-driven basis functions based on the theory of kernel embeddings of conditional distributions developed in the machine learning community. Numerically, we show positive results with two examples. The first example is designed to produce a bimodality in the observation model error (typical of ``cloudy" observations) by introducing obstructions to the observations which occur randomly in space and time. The second example, which is physically more realistic, is to assimilate cloudy satellite brightness temperature-like quantities, generated from a stochastic cloud model for tropical convection and a simple radiative transfer model.}

\begin{document}

\maketitle

\section{Introduction}

Data assimilation \citep[see e.g.,][]{kalnay:03,mh:12} is a sequential method to estimate the conditional distribution of
hidden state variables $x_i$ given noisy observations $y_i$ through Bayes' formula,
\BEA
p(x_i|y_i) \propto p(x_i) p(y_i|x_i),\label{Bayes}
\EEA
where $p(x_i)$ is the prior distribution of the state variables of interest and $p(y_i|x_i)$ is the likelihood function corresponding to the observation model, $y_i = h(x_i) + \eta_i$ at time-$i$. Here, $h$ denotes an observation function and $\eta_i$ represents the measurement noise. In most data assimilation implementations in Numerical Weather Prediction (NWP), one typically assumes that the observation model $h$ is explicitly known.  Moreover, the tacit assumption is that the noise variables $\eta_i$ are independent Gaussian random variables with mean zero and a specified covariance matrix, $R^o$. With these assumptions, the likelihood function is parametrically defined as, $p(y_i|x_i) = p(y_i-h(x_i)) = p(\eta_i) = \mathcal{N}(0,R^o)$. In NWP applications, the prior density $p(x_i)$ at time-$i$ in \eqref{Bayes} is usually represented by an ensemble of forecast solutions of a Global Circulation Model. 

For satellite data assimilation the observations, $y_i$ can be radiances or brightness temperatures measured by satellite instruments. In this particular example, the typical observation model $h$ is the radiative transfer model \citep{liou:02}. While it is believed that high resolution infrared spectral radiances contain detailed information about the temperature and humidity profile, less than 1\% of the AIRS satellite measurements are being used in operational data assimilation problems due to data processing/thinning for quality control purposes and the presence of clouds \citep{realeetal:08}. Assimilating satellite measurement under cloudy conditions is challenging since the presence of clouds in the atmosphere induces significantly cooler radiances (which can be viewed as large biases) relative to those measured under the clear-sky condition. The main challenge in this problem is to detect observation model error that can occur intermittently due to misspecification of the cloud top height, the number of clouds in a column of atmosphere, and/or the cloud fraction in the radiative transfer model of cloudy measurements \citep{mcnally:09}. Mathematically, this suggests that when an incorrect observation model, $\tilde{h}$, is used in placed of the true observation function, $h$, the observations can be written as,
\BEA
y_i = h(x_i) + \eta_i \approx \tilde{h}(x_i) + b_i + \eta_i,\label{obsmodelerror}
\EEA 
where we introduce a biased model error, $b_i$, in addition to the measurement error $\eta_i$. 

In this paper, we introduce a model error estimator to approximate the distribution of the error $b$ at time-$i$, assuming that the underlying observation function $h$ is unknown. Formally, the model error estimator is a Bayesian nonparametric filter which estimates the time dependent posterior density 
\BEA
p(b|y_i)\propto p(b)p(y_i|b). \label{Bayes2}
\EEA
Here, the prior density, $p(b)$, will be constructed based on the predicted observation error. The likelihood function  $p(y_i|b)$ will be constructed nonparametrically using a historical time series of $y$ and $b$, and no knowledge of the true observation function is assumed. Our construction is based on a machine learning tool known as the \emph{kernel embedding of conditional distributions} formulation introduced in \cite{song2009,song2013}. An additional novelty of our approach is that we generalize the formula in \cite{song2009,song2013} to data-driven Hilbert spaces with basis functions obtained from the diffusion maps algorithm \citep{cl:06,bh:16vb}. This is in contrast to their original approach where they specify the Hilbert space using a specific choice of basis functions, such as the radial basis type \citep{song2013}. Once the prior and likelihood are specified, they are combined with \eqref{Bayes2} to form the posterior density $p(b|y_i)$ of the model error.  Finally, we use the statistics of the posterior, such as the mean and variance, to compensate for the error in the observation function and thereby improve the estimation of the state $x_i$. It is important to stress that this model error estimator can be used as a subroutine in tandem with any data assimilation forecasting system based on \eqref{Bayes}. 

We will demonstrate this idea with two synthetic numerical examples. The first example is with the Lorenz-96 model \citep{lorenz:96}, where the observations are corrupted by severe biases at random instances and locations to mimic the bimodality of the observation error distribution in real applications. The second example is with a stochastic cloud model \citep{kbm:10} for tropical convection, where the observed brightness temperature-like quantities are constructed with a simple radiative transfer model \citep{liou:02} with severe biases in the presence of clouds. 

The remainder of this paper is organized as follows. In Section~2, we describe our framework for estimating the model error estimator, using the ensemble Kalman filter (EnKF) as an archetypal example. In Section~3, we discuss the construction of the nonparametric filter in \eqref{Bayes2} which can be combined with any primary filter \eqref{Bayes}. In particular, we describe how to specify the likelihood function and the prior density in \eqref{Bayes2} by combining historical training data with the current observations and primary filter estimates.  In Section~4 we demonstrate our method on the two numerical examples described above and then we briefly conclude in Section~5.

\section{An observation model error estimator}\label{obsmodelerr}

The key issue, as stated in \eqref{obsmodelerror}, is to overcome the observation model error $b$ when we have no access to the true observation function $h$. Let us outline a general framework to mitigate this issue with Fig.~\ref{diagram}. In this diagram, we refer to the data assimilation system which is used to estimate $x$ at time-$i$ in \eqref{Bayes} as the primary filter. Since different operational NWP centers have their preferential methods (4DVAR, EnKF, hybrid, etc), we will design our approach in such a way that it is applicable to any primary filtering scheme. 

Our strategy is to apply a secondary filter in \eqref{Bayes2} and use the resulting posterior conditional statistics to correct the observation model error in the primary filter. In particular, we will use the posterior mean statistics to correct the biases and the variances to correct the additional uncertainties beyond the measurement error. We should stress that the implementation of the secondary filter in this framework offers no additional changes to the infrastructure in the primary filter except for a simple two-way communication at each assimilation step. Namely, one needs to feed the predicted observation error from the primary filter into the secondary filter and then feed the model error statistics from the secondary filter back into the primary filter (to correct the likelihood functions in the primary filter). The third row in the diagram in Fig.~\ref{diagram} clarifies that we use the current observations to construct the prior density and the likelihood function used in the secondary filter. Now let us illustrate this heuristic discussion with a concrete algorithm in the case where the primary filter is the ensemble Kalman filter (EnKF).

\begin{figure}
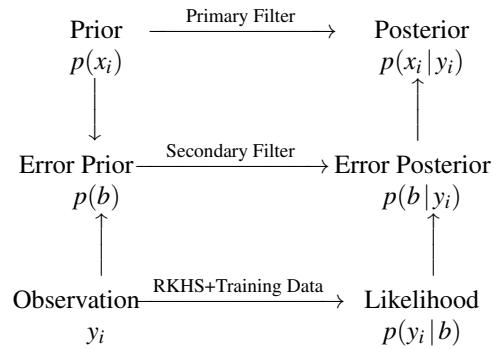

\[ \arraycolsep=1.4pt\begin{array}{clc}
\textup{Prior}  \hspace{0pt}&\xrightarrow{\ \ \ \ \ \textup{Primary Filter} \ \ \ \ \ }& \textup{Posterior} \\
p(x_i) & & p(x_i \, | \, y_i) \\
\hspace{0pt}\left\downarrow\rule{0cm}{.5cm}\right.   \scriptstyle{ }&&\hspace{0pt}\left\uparrow\rule{0cm}{.5cm}\right. \scriptstyle{\textup{ }}\\ 
\hspace{-15pt}\textup{Error Prior} &\hspace{-5pt}\xrightarrow{\ \ \ \ \textup{Secondary Filter} \ \ \ \ }&\hspace{-7pt}\textup{Error Posterior} \\ 
 p(b) & & p(b \, | \, y_i) \\
\hspace{5pt}\left\uparrow\rule{0cm}{.5cm}\right.   \scriptstyle{ }&&\hspace{10pt}\left\uparrow\rule{0cm}{.5cm}\right. \scriptstyle{ }\\
\hspace{-15pt}\textup{Observation} \hspace{0pt}&\hspace{-5pt}\xrightarrow{\ \ \textup{RKHS+Training Data} \ \ }& \textup{Likelihood}  \\
y_i & & p(y_i\, | \, b)
\end{array} \]
\caption{Diagram representing the integration of the secondary filter into an arbitrary primary filtering procedure.}
\label{diagram}
\end{figure}

With the EnKF, both the prior and posterior distributions of the primary filter are assumed to be Gaussian and only the first two empirical moments are corrected using the Kalman filter formulas.  The Kalman filter implicitly assumes that the observation error is unbiased, $\mathbb{E}[y_i-h(x_i)] =\mathbb{E}[\eta_i]  = 0$.  In the $i$-th analysis step, an ensemble of $K$ prior estimates $x^{k,f}_i\sim \mathcal{N}(\bar{x}^f_i,P^f_{xx,i})$, where $k\in \{1,...,K\}$ denotes the ensemble member, is transformed into an ensemble of analysis estimates $x^{k,a}_i \sim \mathcal{N}(\bar{x}^a_i, P^a_{xx,i})$. Here the superscripts $f$ and $a$ are to denote estimates before (prior/forecast) and after (analysis) the observation $y_i$ has been assimilated, respectively. In each analysis step, the EnKF procedure can be summarized in three steps. First, it approximates the prior mean and covariance statistics with empirically estimated statistics defined as,
\BEA
\bar{x}^{f}_i = \frac{1}{K}\sum_{k=1}^K x_i^{k,f}, \quad P_{xx,i}^{f} &=& \frac{1}{K-1}X_iX _i^\top + Q, 
\EEA
where $X_i = [x_i^{1,f} - \bar{x}^f_1,\ldots,x_i^{K,f} - \bar{x}^{f}_i]$. Second, it applies the Kalman filter formula to obtain the posterior mean and covariance,
\BEA
    \bar{x}^{a}_i &=&  \bar{x}^f_i+K_i(y_i- \bar{y}^f_i),  \nonumber\\
    P_{xx,i}^a &=& P_{xx,i}^f - K_i P_{xy,i}^\top, \label{enkf}\\
K_i &=& P_{xy,i}(P_{yy,i}+R)^{-1}.  \nonumber
\EEA
Third, it draws an ensemble of analysis estimates $x^{k,a}_i \sim \mathcal{N}(\bar{x}^a_i, P^a_{xx,i})$. In \eqref{enkf},   
$\bar{y}^f_i = {K}^{-1}\sum_{k=1}^K y^{k,f}_i$, $P_{xy,i} = (K-1)^{-1}X_iY_i^\top$ and $P_{yy,i}= (K-1)^{-1}Y_iY_i^\top$ are the empirically estimated statistics, where $Y_i = [y_i^{1,f} - \bar{y}^f_1,\ldots,y_i^{K,f} - \bar{y}^{f}_i]$, and $y^{k,f}_i = h ( x^{k,f}_i )$ are the ensemble of the predicted observations. Several methods to execute \eqref{enkf} and to draw the analysis ensembles have been introduced, for example, with perturbed observations \citep{evensen:94}, or with square root filters \citep{bishop:01,anderson:01}.  Notice that this method depends on two parameters $Q$ and $R$ which represent the covariance matrices of the dynamical and observation noise.  In particular, when the observation model is perfect we set $R=R^o$ meaning that the only observation noise is the measurement error.  

We assume that the true observation model $h$ is unknown and we are given the imperfect observation model $\tilde{h}$. Define $b_i := h(x_i) - \tilde{h}(x_i)$ as the observation model error at the $i$-th time step and assume that the observation model error is biased, that is, $\mathbb{E}[b_i] = \mu_{b_i} \neq 0$ and independent to the measurement noise, $\eta_i$. With this configuration, it is clear that if we consider filtering with the imperfect observation model in \eqref{obsmodelerror}, the observation error is biased, $\mathbb{E}[y_i-\tilde{h}(x_i)] = \mu_{b_i}\neq 0$.

One way to mitigate this issue is by adjusting the observation model as follows,
\BEA
y_i - \mu_{b_i} = \tilde{h}(x_i) + \tilde{b}_i + \eta_i, \quad\quad \eta_i\sim\mathcal{N}(0,R),\label{unbiasedobsmodel}
\EEA
where we define $\tilde{b}_i := b_i-\mu_{b_i}$ as the unbiased observation model error with covariance matrix $R_{b_i}:= \mathbb{E}[\tilde{b}_i\tilde{b}_i^\top]$. Assuming that $\tilde b_i$ and $\eta_i$ are independent,  we can implement the EnKF with the unbiased observation model in \eqref{unbiasedobsmodel} as follows,
\BEA
\bar{x}^{a}_i &=&  \bar{x}^f_i+\tilde{K}_i(y_i-\mu_{b_i} - \tilde{y}^f_i),  \nonumber\\
\tilde{P}_{xx,i}^a &=& \tilde{P}_{xx,i}^f - \tilde{K}_i \tilde{P}_{xy,i}^\top, \label{unbiasedEnKF}\\
\tilde{K}_i &=& \tilde{P}_{xy,i}(\tilde{P}_{yy,i}+R+R_{b_i})^{-1}.  \nonumber
\EEA
Here, $\tilde{y}^f_i = {K}^{-1}\sum_{k=1}^K \tilde{y}^{k,f}_i$, $\tilde{P}_{xy,i} = (K-1)^{-1}X_i\tilde{Y}_i^\top$, $\tilde{P}_{yy,i}= (K-1)^{-1}\tilde{Y}_i\tilde{Y}_i^\top$, where $\tilde{Y}_i = [\tilde{y}_i^{1,f} - \tilde{y}^f_1,\ldots,\tilde{y}_i^{K,f} - \tilde{y}^{f}_i]$ and $\tilde{y}^{k,f}_i = \tilde{h}\left( x^{k,f}_i \right)$. In order to implement the unbiased filter in \eqref{unbiasedEnKF}, we need to estimate $\mu_{b_i}$ and $R_{b_i}$. We propose to use the conditional statistics of the secondary filter in \eqref{Bayes2}, $\hat{\mu}_{b_i}=\mathbb{E}[b|y_i]$ and $R_{b_i}=Var[b|y_i]$, as the estimators for $\mu_{b_i}$ and $R_{b_i}$, respectively. With this goal in mind, we now explain the secondary filter.

\section{The secondary filter}
The goal of this section is to explain the construction of the nonparametric filtering in \eqref{Bayes2}. Since this requires various technical tools from the machine learning community, we accompany the discussion with several Appendices for detailed discussions. 

We assume that we are given only pairs of data points $\{(x_\ell,y_\ell)\}_{\ell=1}^N$ for training (so no knowledge of the true observation function $h$ is assumed).  From these data points we compute the implied model error $b_\ell=y_\ell-\tilde{h}(x_\ell)$, where $\tilde{h}$ is the given imperfect observation model. First, we will discuss how to train the likelihood function using this dataset. Second, we will discuss how to extend the likelihood function on new observations to obtain the conditional probability of $b$ given new observations $y_i \notin \{y_\ell\}_{\ell=1}^N$ that are not part the training dataset. Third, we will discuss how to construct a prior. Subsequently, we discuss how to extract the estimators $\mu_{b_i}$ and $R_{b_i}$ that are needed in \eqref{unbiasedEnKF}.

\subsection{Training the nonparametric likelihood function}

Our goal here is to learn the likelihood function $p(y|b)$ by using a training data set consisting of pairs $\{(b_\ell,y_\ell)\}_{\ell=1}^N$. The outcome of this training is a matrix of size $N\times N$, where the $(i,j)$-entry is an estimate of the conditional density, $p(y_{i}|b_{_j})$.  This matrix is a discrete representation of the likelihood function evaluated at each point of the training dataset $y_i \in \{y_\ell\}_{\ell=1}^N$ and $b_j \in \{b_\ell\}_{\ell=1}^N$. This matrix is a nonparametric representation of $p(y|b)$.  However, since the data set could be quite large, we will never explicitly construct this matrix, and instead it will be represented by its projection into a lower-dimensional set of basis elements.

For clarity of presentation, we assume in this section that the observation is one dimensional. In order to correct a multidimensional observation, the following bias correction steps are applied independently to each coordinate of the observation. With this scalar formulation, the computational cost in constructing the likelihood function becomes manageable and the secondary filter is easily parallelize-able for high-dimensional observations.

Our main idea is to represent the conditional density $p(y|b)$ with a set of basis functions which can be learned from the training data set using the diffusion maps algorithm \citep{cl:06,bh:16vb}. First let us discuss how to construct the basis functions. Subsequently, we discuss a nonparametric representation of conditional density functions.

\subsubsection{Learning the data-driven basis functions}
In a nutshell, the diffusion maps algorithm can be described as follows. Given a dataset $x_\ell\in\mathcal{M}\subset\mathbb{R}^n$ with sampling density $q(x)$, defined with respect to the volume form inherited by the manifold $\mathcal{M}$ from the ambient space $\mathbb{R}^n$, the diffusion maps algorithm is a kernel-based method which constructs an $N\times N$ matrix $L$ that approximates a weighted Laplacian operator, $\mathcal{L} = \nabla \log(q)\cdot\nabla +\Delta$. The eigenvectors $\vec{\varphi}_k$ of the matrix $L$ are discrete estimates of the eigenfunctions $\varphi_k(x)$ of the operator $\mathcal{L}$ which form an orthonormal basis of a weighted Hilbert space $L^2(\mathcal{M},q)$. 

For example, if the data is uniformly distributed, $q(x)=1$, on the unit circle $\mathcal{M}=S^1$, then the matrix $L$ approximates the Laplacian $\Delta$ on this periodic domain. In this case, eigenvectors of $L$ approximates eigenfunctions of the Laplacian operator on the unit circle, which are the Fourier functions that form an orthonormal basis of $L^2(S^1)$. Thus one can think of the diffusion maps algorithm as a method to specify a generalized Fourier basis adapted from the data. The basis functions $\varphi_k(x)$ are represented nonparametrically by the vectors $\vec{\varphi}_k\in\mathbb{R}^N$ whose $\ell$-th component is a discrete estimate of the eigenfunction $\phi_k(x_\ell)$, evaluated at training data point $x_\ell$.

In our application, we apply the diffusion maps separately on the dataset $b_\ell \in\mathbb{R}$ and $y_\ell \in\mathbb{R}$. Let $q(b)$ and $\tilde{q}(y)$ be the sampling densities of $b_\ell$ and $y_\ell$, respectively. Implementing the diffusion maps algorithm, we obtain vectors $\vec{\varphi}_k$ which approximate $\varphi_k(b)\in L^2(\mathbb{R},q)$ and $\vec{\phi}_k$ which approximate $\phi_k(y) \in L^2(\mathbb{R},\tilde{q})$. In our implementation, we use the variable bandwidth kernels introduced by \cite{bh:16vb}. We refer to the Appendix in \citet{bh:16physd} for the pseudo-code of the algorithm. 

\subsubsection{A nonparametric representation of conditional density functions}

Let $\varphi_j(b)\in L^2(\mathbb{R},q)$ and $\phi_k(y) \in L^2(\mathbb{R},\tilde{q})$ be the basis functions approximated by the diffusion maps algorithm. For finite modes, $j=1,\ldots, M_2, k=1,\ldots, M_1$, we consider a nonparametric representation of the conditional density as follows,
\BEA
p(y|b) = \sum_{k=1}^{M_1} \mu_{Y|b,k} \phi_k(y) \tilde q(y).  \label{pycondx}
\EEA
where the expansion coefficients are defined as follows,
\BEA
 \mu_{Y|b,k} = \sum_{j=1}^{M_2}  \varphi_j(b) [C_{YB}C_{BB}^{-1}]_{kj}.\label{muyb}
\EEA
Here matrices $C_{YB}$ is $M_2 \times M_1$  and $C_{BB}$  is $M_1\times M_1$  whose components can be approximated by Monte-Carlo averages as follows,
\BEA
\big[C_{YB}\big]_{jk} &\approx&  \frac{1}{N} \sum_{\ell=1}^N \phi_j(y_\ell) \varphi_k(b_\ell) \label{CYB}\\
\big[C_{BB}\big]_{jk} &\approx& \frac{1}{N} \sum_{\ell=1}^N \varphi_j(b_\ell) \varphi_k(b_\ell)\label{CBB}.
\EEA
The equation for the expansion coefficients in \eqref{muyb} is based on the theory of kernel embedding of conditional distribution introduced in \cite{song2009,song2013} which we reviewed in Appendix A below. See Appendix B for the detailed proof of equations \eqref{muyb}-\eqref{CBB}.  

From the expression in \eqref{muyb}, one can see that the conditional density in \eqref{pycondx} is represented as a regression in infinite dimensional spaces with basis functions of $\varphi_j(b)$ and $\phi_k(y)$.
 This representation is nonparametric in the sense that we do not assume that the density function is of particular distribution. Numerically, the nonparametric representation of $p(y|b)$ is given by an $N\times N$ matrix whose $(i,j)$th component is 
\BEA
p(y_i|b_j) = \sum_{k=1}^{M_1} \mu_{Y|b_j,k} \phi_k(y_i) \tilde q(y_i),\label{matrixpb|y}
\EEA
where $y_i \in \{y_\ell\}_{\ell=1}^N$ and the coefficients are given by \eqref{muyb} evaluated at the training data $b_j \in \{b_\ell\}_{\ell=1}^N$. From \eqref{matrixpb|y}, notice that all we need are the function values $\phi_k(y_\ell)$, $\varphi_j(b_\ell)$, $\tilde{q}(y_\ell)$, which are obtained via the diffusion maps algorithm. We should note that the sampling density $\tilde{q}$ is estimated using a kernel density estimation method in our implementation of the diffusion maps algorithm (see Appendix of \cite{bh:16physd} for the algorithmic detail). In our numerical implementation, we will set $M_1=M_2=M$. 

\subsection{Extension of the likelihood function to new observations}

To construct $p(y_i\,|\,b_\ell)$ for $y_i \neq\{y_\ell\}_{\ell=1}^N$ that is not in the training data set, Eq.~\eqref{pycondx} suggests that we need to extend the eigenfunctions to this new data point, $\phi_k(y_i)$. One approach would be to use the Nystr\"{o}m extension \citep{nystrom:30}, however, since the observation $y_i$ is noisy, a more robust method is to compute weights
\begin{align} p(y_\ell \, | \, y_i) = \frac{1}{\sqrt{2\pi R}}\exp\left(-\frac{(y_\ell - y_i)^2}{2R}\right)\nonumber \end{align}
which are the probabilities of the training data points given the observation $y_i$.  We then estimate 
\begin{align} \mathbb{E}_{\eta_i}[\phi_k(y_i)\tilde q(y_i)] &= \mathbb{E}_{p(\cdot \, | \, y_i)}[\phi_k(\cdot)\tilde q(\cdot)] \nonumber \\
 &\approx \frac{1}{N}\sum_{\ell=1}^N p(y_\ell \, | \, y_i)\phi_k(y_\ell)\tilde q(y_\ell)\nonumber \end{align}
and use this expectation on \eqref{pycondx} to define the likelihood function, 
\begin{align}\label{biasLikelihood}
p(y_i \, | \, b_\ell)  &= \mathbb{E}_{\eta_i}[ p(y_\ell \, | \, b_\ell)]  \nonumber \\ &=  \sum_{k=1}^M  \mu_{Y| b_\ell,k}\mathbb{E}_{\eta_i}[ \tilde \varphi_k(y_i) \tilde q(y_i)],
\end{align}
where the coefficients $\mu_{Y|b_\ell,k}$ are defined in \eqref{muyb} and computed in the training phase. To conclude, we represent the likelihood function $p(y_i|b)$ nonparametrically by an $N$-dimensional vector whose $\ell$th component is $p(y_i \, | \, b_\ell)$ as prescribed in \eqref{biasLikelihood}. In the remainder of this paper, we denote this nonparametric likelihood function as the \emph{RKHS likelihood function}.

\subsection{Prior density functions}

An appropriate construction of the prior distribution $p(b)$ at each data assimilation time $i$ is essential for accurate filter estimates.  This is especially true when the likelihood function $p(y_i|b)$ is bimodal as a function of $b$, as we will see in the numerical applications below. In this article, we consider a Gaussian prior distribution with mean and variance given by,
\BEA
\hat b_i = y_i - \bar{y}_i^f,\quad\quad  \sigma_{b_i}^2 = P_{yy,i} + R.\label{sigmab} 
\EEA 
Given a training data set $\{(x_\ell,y_\ell)\}_{\ell=1}^N$, we can compute the training biases $b_\ell = y_\ell - \tilde h(x_\ell)$ such that the prior density evaluated at the training data is given by,
\BEA
\label{biasPrior} p(b_\ell) \propto \exp\left(-\frac{(b_\ell-\hat b_i)^2}{2\sigma_{b_i}^2}\right).
\EEA
So, $p(b)$ is represented by an $N$-dimensional vector whose $\ell$th component is $p(b_\ell)$ as defined in \eqref{biasPrior}. As an alternative to the time dependent variance defined in \eqref{sigmab}, one can also use the climatological variance obtained from the historical training $b_\ell$ (this is analogous to a 3DVAR prior). While more complicated priors are always possible, we will apply our numerical experiments below using these Gaussian prior densities.

\subsection{Statistics of the posterior distribution} Now that we have defined the likelihood $p(y_i \, | \, b_\ell)$ in \eqref{biasLikelihood} and the prior $p(b_\ell)$ in \eqref{biasPrior}, we can multiply these terms to form the posterior density,
\begin{equation}
p(b_\ell \, | \, y_i) = \frac{1}{Z}  p(b_\ell)p(y_i \, | \, b_\ell), \nonumber
\end{equation}
where 
\[ Z = \int p(b)p(y_i \, | \, b) \, db \approx \frac{1}{N}\sum_{\ell=1}^N p(b_\ell)p(y_i \, | \, b_\ell)q(b_\ell)^{-1} \] 
is the normalization constant estimated via Monte-Carlo summation. With this posterior density estimate, we can compute the posterior mean and variance,
\begin{align}\label{postStats}
\hat{\mu}_{b_i} &= \frac{1}{N} \sum_{\ell=1}^N b_\ell p(b_\ell \, | \, y_i)q(b_\ell)^{-1} \nonumber \\ \hat{R}_{b_i} &= \frac{1}{N} \sum_{\ell=1}^N (b_\ell-\tilde b_i)^2 p(b_\ell \, | \, y_i)q(b_\ell)^{-1}.
\end{align}
as estimators for $\mu_{b_i}$ and $R_{b_i}$ that can be used in the primary filtering step in \eqref{unbiasedEnKF}. This completes our description of the secondary filter. In the remainder of this paper, we denote the EnKF in \eqref{unbiasedEnKF} with observation error corrected using the statistics in \eqref{postStats} as the \emph{RKHS filter}.

\section{Example 1: Assimilating random ``cloudy'' observations} 

\begin{figure}
\centering
\includegraphics[width=0.48\linewidth]{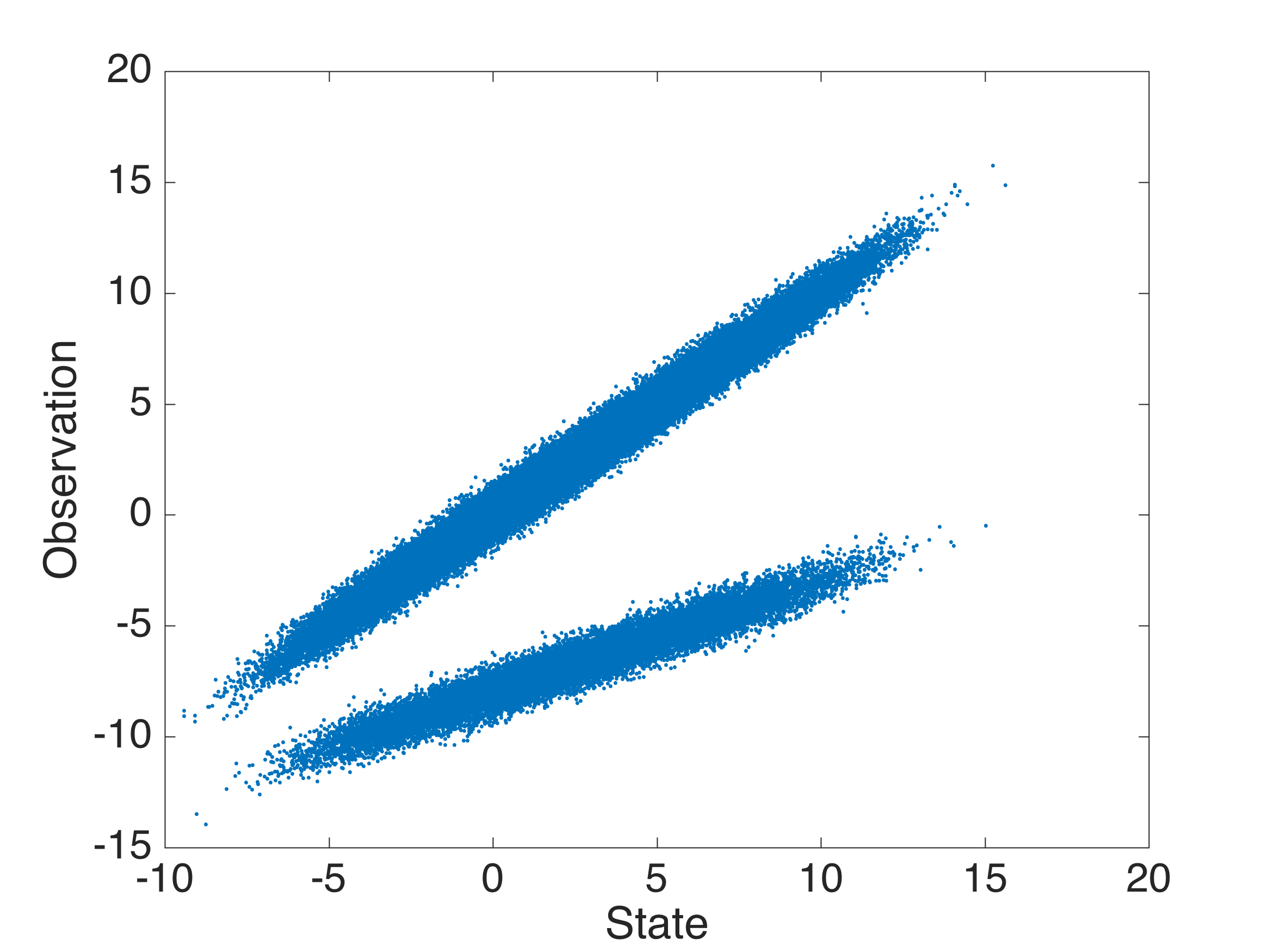}
\includegraphics[width=0.48\linewidth]{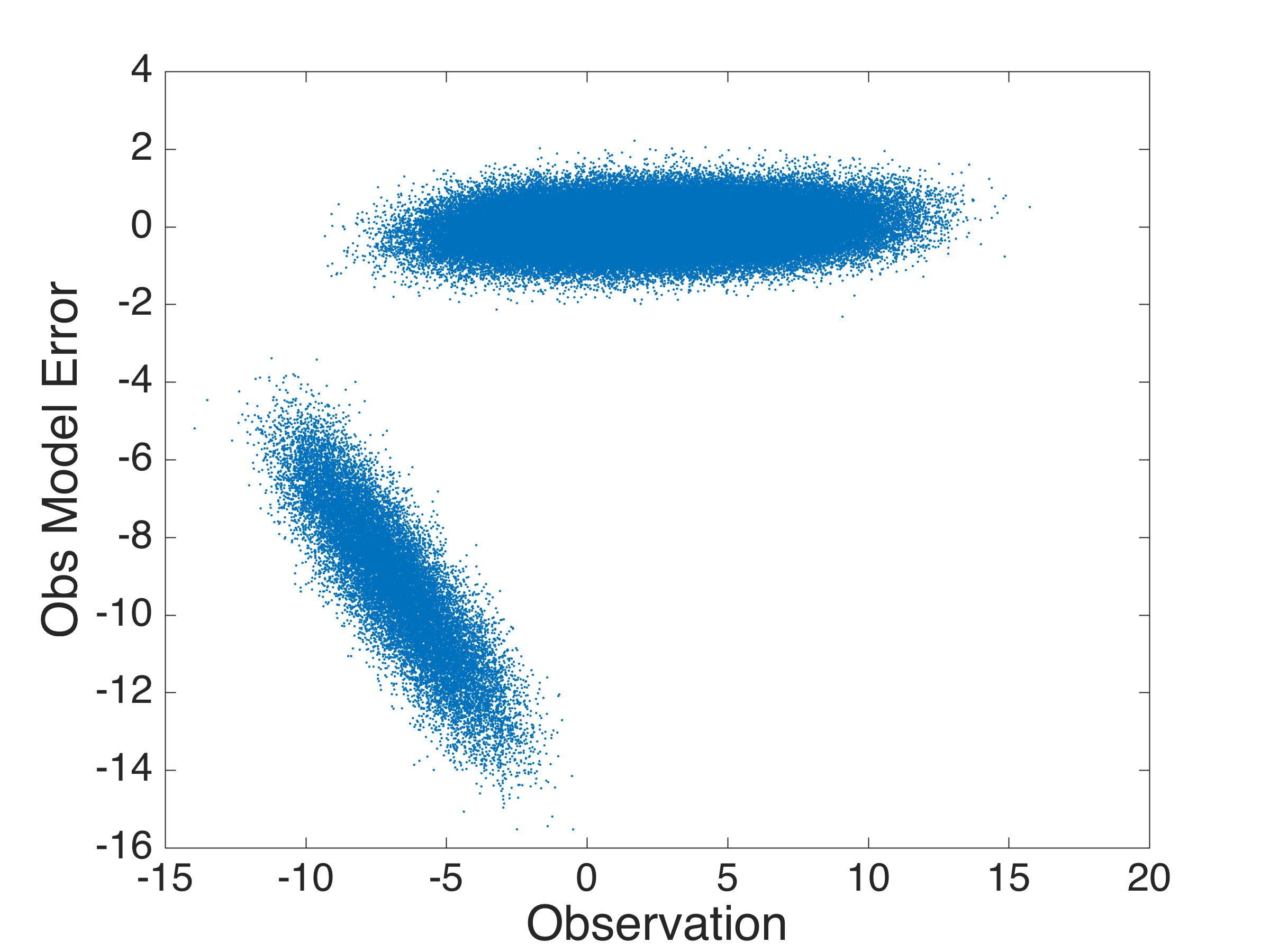}
\includegraphics[width=0.48\linewidth]{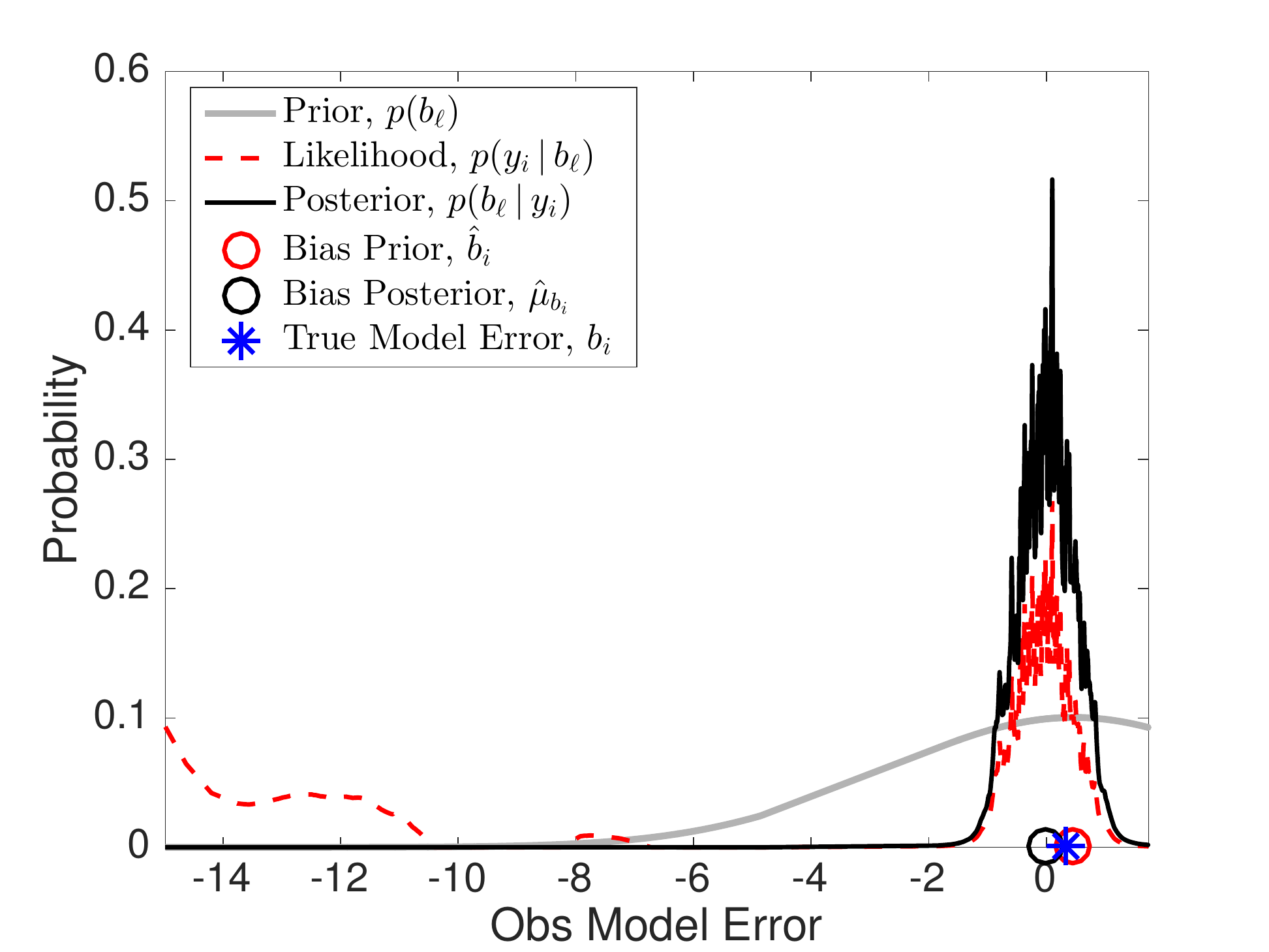}
\includegraphics[width=0.48\linewidth]{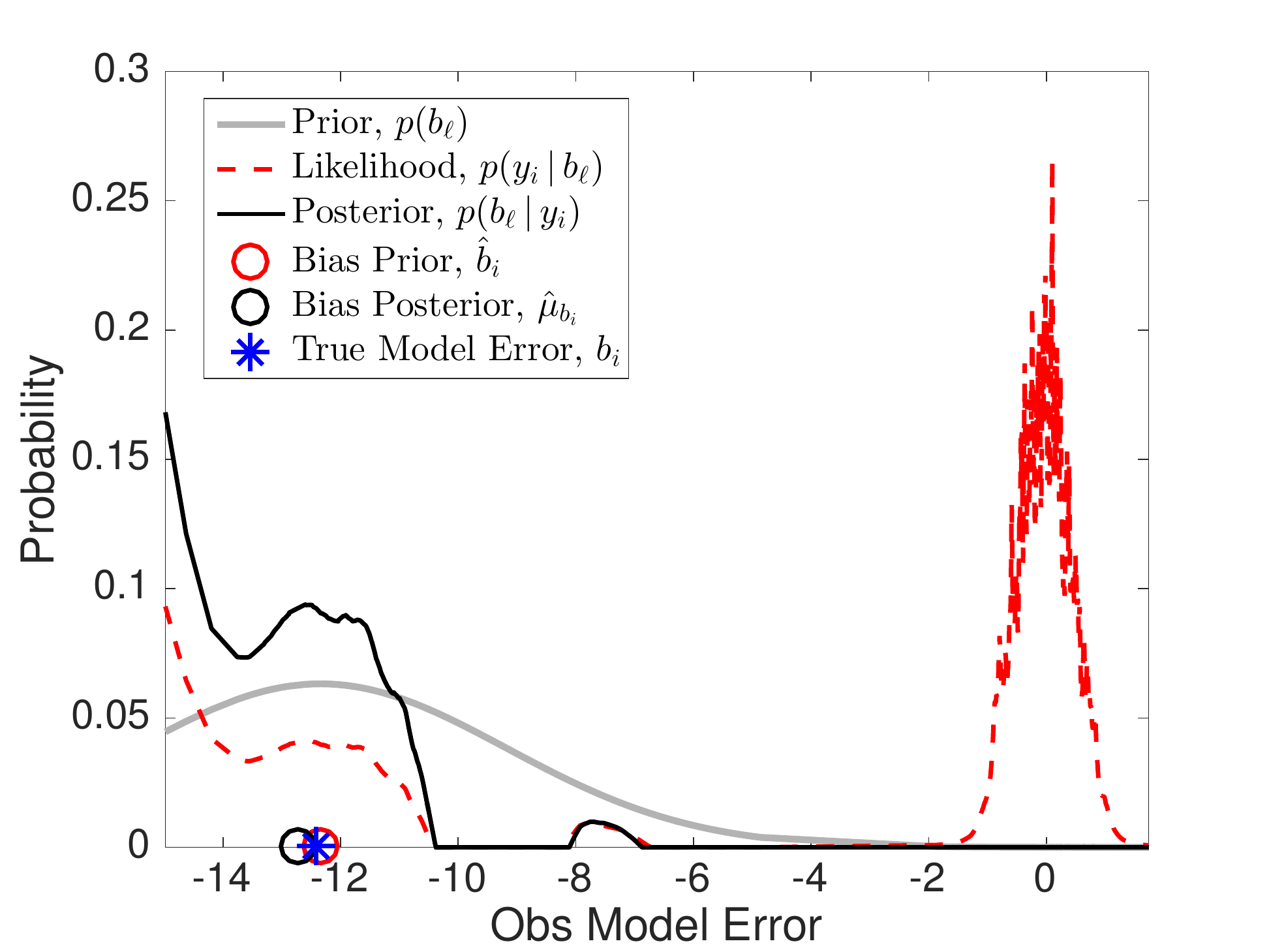}
\caption{\label{obsvsstate} Top, Left: Visualization of the observations using \eqref{cloudyObs} plotted against the true state, the top ellipse corresponds to the clear-sky observations and the bottom ellipse corresponds to the cloudy observations. Top, Right: Visualization of the observation model error as a function of the observation.  Bottom, Left: Example of a bimodal RKHS likelihood function (red, dashed) indicating the current observations maybe obstructed or not obstructed, the prior \eqref{biasPrior} derived from the primary filter (gray, solid) reveals it to be unobstructed as shown in the posterior (black, solid).  Bottom, Right: Similar bimodal likelihood but the prior indicates obstructed observations.}
\end{figure}

In this example we will introduce a severe error into the observations of the 40-dimensional chaotic \citet{lorenz:96} system,
\begin{eqnarray}\label{l96}
\dot x_j = x_{j-1}(x_{j+1}-x_{j-2}) - x_j + 8, 
\end{eqnarray}
where the indices $j$ are taken modulo 40. The underlying truth is generated by the RK4 integration scheme with integration time step $\delta t=0.05$ and a randomly chosen initial condition. Direct observations are taken at every discrete time step $\Delta t=t_{i+1}-t_i$, which we will vary between 0.1-0.5. We will show results for observing $x_j$ at every other grid points (20 observations). At each observation time step $t_i$, we randomly choose up to $c=7$ from the 20 observed locations and obstruct the corresponding observations with a ``cloud'' as follows. We draw $\xi_i\sim \mathcal{U}(0,1)$ and let the observations at these $c$ locations be,
\begin{align}\label{cloudyObs} 
h(x_k) &= \left\{ \begin{array}{ll} x_k  & \xi_i > 0.8  \\ \beta_k x_k - 8 & \textup{ else } \end{array} \right.\\
\beta_k &\sim \mathcal{N}(0.5,1/50). \nonumber
\end{align}
Effectively, up-to seven random locations out of 20 direct measurement have an 80\% chance of being randomly scaled by $\beta_k$ and shifted down by eight units. In our experiment, the observed data is generated with this hidden observation model and the filter observation model is $\tilde{h}(x_k)=x_k$. Therefore, when $\xi_i>0.8$, we directly observe the state and the observation model is correct and we refer to this case as an \emph{unobstructed} (``clear-sky") observation. On the other hand, when $\xi_i\leq 0.8$, the filter observation model is incorrect, $\tilde{h}\neq h$, and we refer to this case as an \emph{obstructed} (``cloudy") observation. Notice that at each time there are up to $7$ randomly chosen ``cloudy" locations. To represent the measurement error, we also include additive Gaussian noise $\eta_k \sim \mathcal{N}(0,R^o)$ in the observations $y_k = h(x_k) + \eta_k$.  Since the obstructions in the observation $h$ appear randomly, they are impossible to predict, and we assume that the modeler only has the incorrect observation model $\tilde h(x_k) = x_k$.  For comparison, we also generate a set of `clear sky' observations $\tilde y_k = \tilde h(x_k) + \eta_k = x_k + \eta_k$ so that we can test the standard EnKF in the case when there is no model error.

In Fig.~\ref{obsvsstate} we show a scatter plot which clearly illustrates the bimodal distribution of the observations, $y_k$.  Notice that the clear-sky and cloudy observations are both permitted in the range $[-9,-2]$ making it impossible to tell from the observation alone whether the observation is obstructed or not. The conditional density, $p(y_i\,|\,b_\ell)$ is trained from a short training data set $\{(b_\ell,y_\ell)\}_{\ell=1}^N$, where $N=10000$. In this example since model error is spatially homogeneous, we fit a single model and use it for each of the $20$ observations. This assumption also allows us to use only $500$ time steps of training data (since each time step contains $20$ observations, this short time series gives us $N=10000$ observations). In this numerical experiment, the likelihood function is constructed using $M=250$ eigenfunctions. 

In the bottom panels of Fig.~\ref{obsvsstate}, we show examples of the RKHS likelihood functions (red, bottom left and right) both of which result from observations near $y_i \approx 4$, but in one case the observation is obstructed (bottom left) and in the other the observation is unobstructed (bottom right). In these panels, we also show the Gaussian prior distribution as proposed in \eqref{biasPrior} with a time-independent $\sigma_b^2$ obtained by taking the variance of the training data $b_\ell$. Notice that the prior helps the secondary filter to identify whether the observation error is small (bottom, left panel) or large (bottom,right). In this case, the climatological prior variance is sufficient to give a relatively informative prior. In the next example, we will need a time-dependent prior since the variance is too broad. 

In Fig.~\ref{timetrace}, we show the spatiotemporal structure of the observations, the true state, and the EnKF without and with the observation model error corrections. Notice that the observations are severely corrupted (see the deep blue dots) and these obstructions occur completely randomly in space and time. Note also that the ``cloudy'' observations cause observation model error because we assume no knowledge of the true observation function $h$, and we filter using $\tilde{h}(x_k) = x_k$ on each observed location. Here, the primary filter is implemented with 80 ensemble members and an adaptive covariance estimation of the system and observation noise covariances, $Q$ and $R$ (see \cite{bs:13} for details).  Below, we will also include results without using the adaptive covariance estimation for comparison.  We use the abbreviation EnKF to refer to the primary filter without correcting the observation model error, that is, $\mu_{b_i}=0$ and $R_{b_i}=0$ in \eqref{unbiasedEnKF} and we use the abbreviation RKHS to refer to the EnKF with the observation model error corrected by \eqref{postStats}.

In Fig.~\ref{RandDT}, we show the Root-Mean-Square Error (RMSE) between the truth and the filtered estimates as functions of measurement error, $\sqrt{R^o}$, and observation time, $\Delta t$.  Each RMS error is averaged over $5000$ filter steps after removing an initial transient to allow each filter to converge.  In each panel, we include the results of applying the EnKF to unobstructed (clear sky) observations $\tilde y_k$ (black, solid line) which should be considered the best possible filtering since the observation model is correct when the observations are all unobstructed.  On the other hand, when the observations are obstructed, $y_k$ (cloudy sky) the observation model is incorrect and the results (red, dotted line) degrade significantly (notice that gaps in the curve indicate catastrophic filter divergence). 

The results shown in the top row of Fig.~\ref{RandDT} correspond to implementing the filter using the true value $R=R^o$ and a fixed small additive inflation $Q=10^{-3}\times I_{20\times 20}$. In this case, the classical EnKF filter estimates diverge in the presence of this severe observation model error. This result suggests that the empirical additive inflation (which is commonly used to compensate for model error in applications) is not able to overcome the ill-posedness induced by the bimodality of the model error distribution in this example. In the bottom row, we include results which use $Q$ and $R$ parameters determined by the adaptive covariance estimation method \citep{bs:13}. Notice that the EnKF applied to the cloudy observations with the adaptive estimation of $Q$ and $R$ is able to prevent catastrophic filter divergence except when the observation time $\Delta t>0.3$. For both the fixed parameter and adaptive filters, it is clear that the RKHS filter given the cloudy observations approaches the performance of the EnKF given the clear sky observations in the limit of small measurement error and small observation time.

In Fig.~\ref{EstimatedQR} we plot the average of the diagonal entries of the $Q$ and $R$ parameters estimated by the adaptive covariance estimation as functions of both the measurement error covariance and the observation time. First, the EnKF with clear-sky observations recovers the true values of $R$ with high accuracy when the observation time is relatively small ($\Delta t<0.35$). Notice that when the EnKF is given the cloudy observations, the estimates for $Q$ and $R$ are much larger than the estimates than when given the clear-sky observations. Moreover, the estimates from the RKHS filter are closer to those of the EnKF given the clear-sky observations. Intuitively, this means that the RKHS model error correction is effective and that the adaptive method only needs to compensate for the remaining measurement error which is present even in the clear sky observations.  In other words, the RKHS correction is effectively removing the clouds from the cloudy observations and achieving results close to the EnKF given the clear sky observations.


\begin{figure}
\centering
\includegraphics[width=0.95\linewidth]{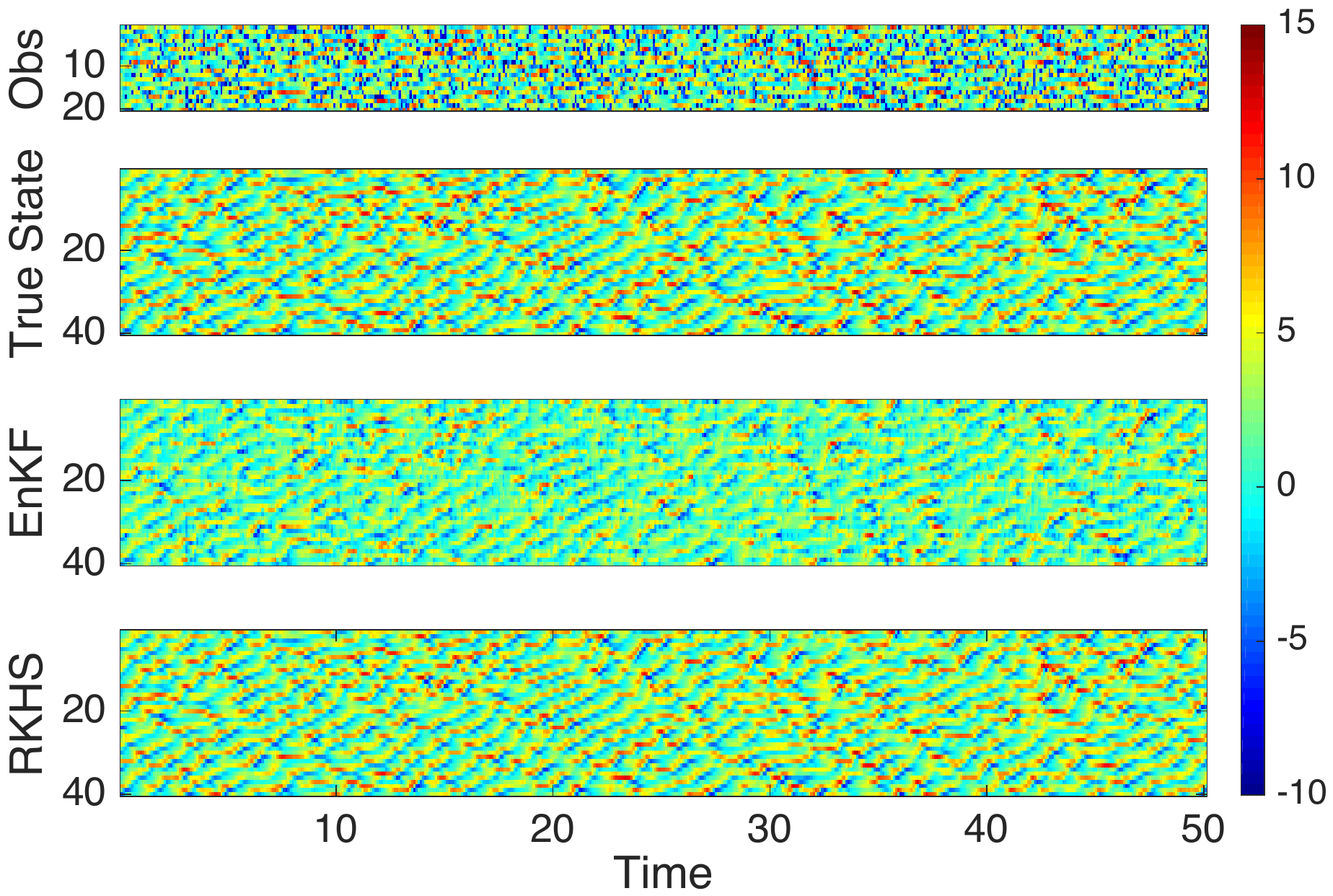}
\caption{\label{timetrace} The spatiotemporal patterns of the sparse observations (top) simulated using \eqref{cloudyObs} with $c=7$ randomly located obstructed observation, compared with the true state (second row) and the filtered state estimates of the EnKF with adaptive tuning of Q and R (third row) and the RKHS corrected EnKF (bottom).}
\end{figure}

\begin{figure}
\vspace{-10pt}
\centering
\includegraphics[width=0.48\linewidth]{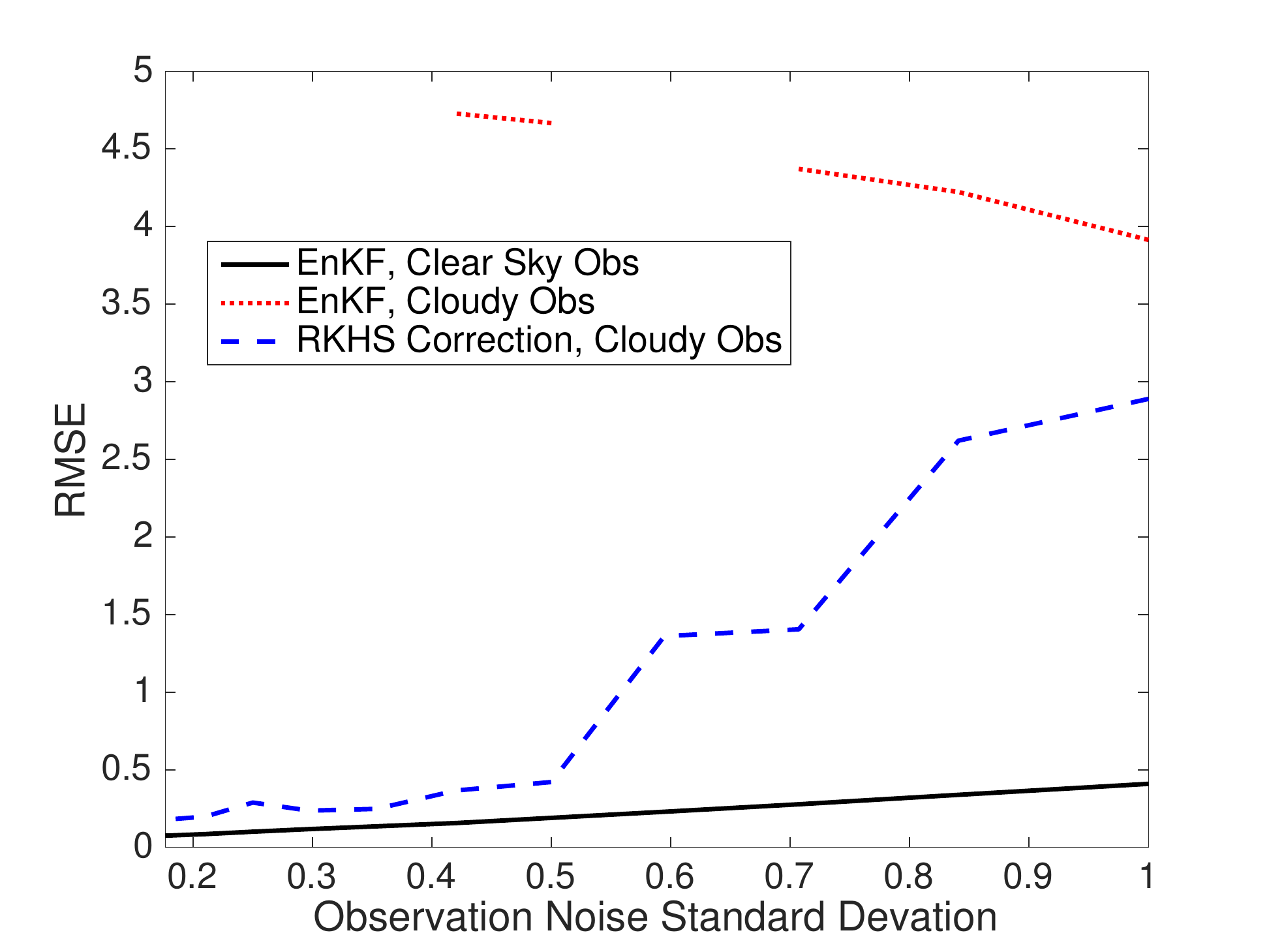}
\includegraphics[width=0.48\linewidth]{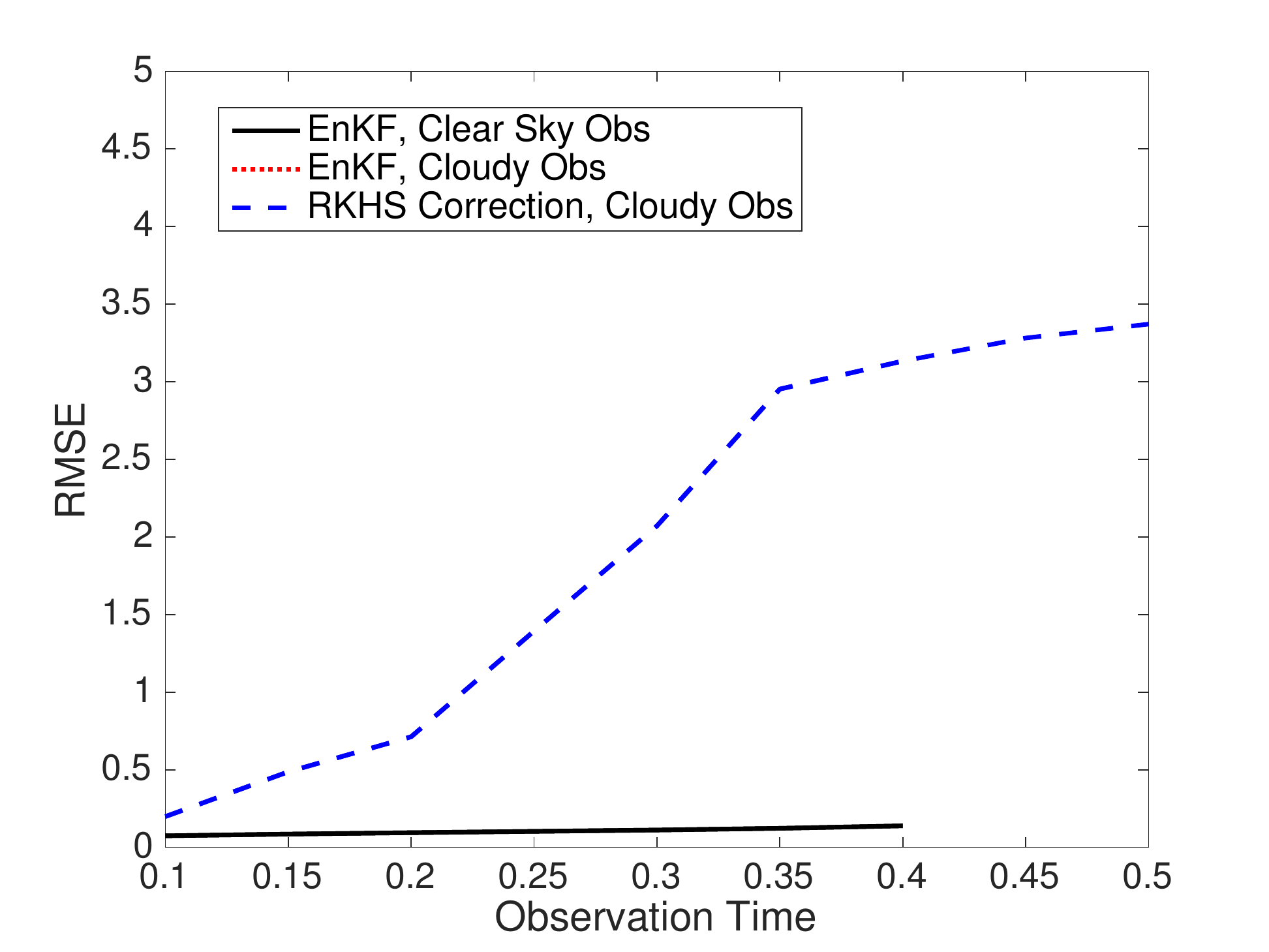}
\includegraphics[width=0.48\linewidth]{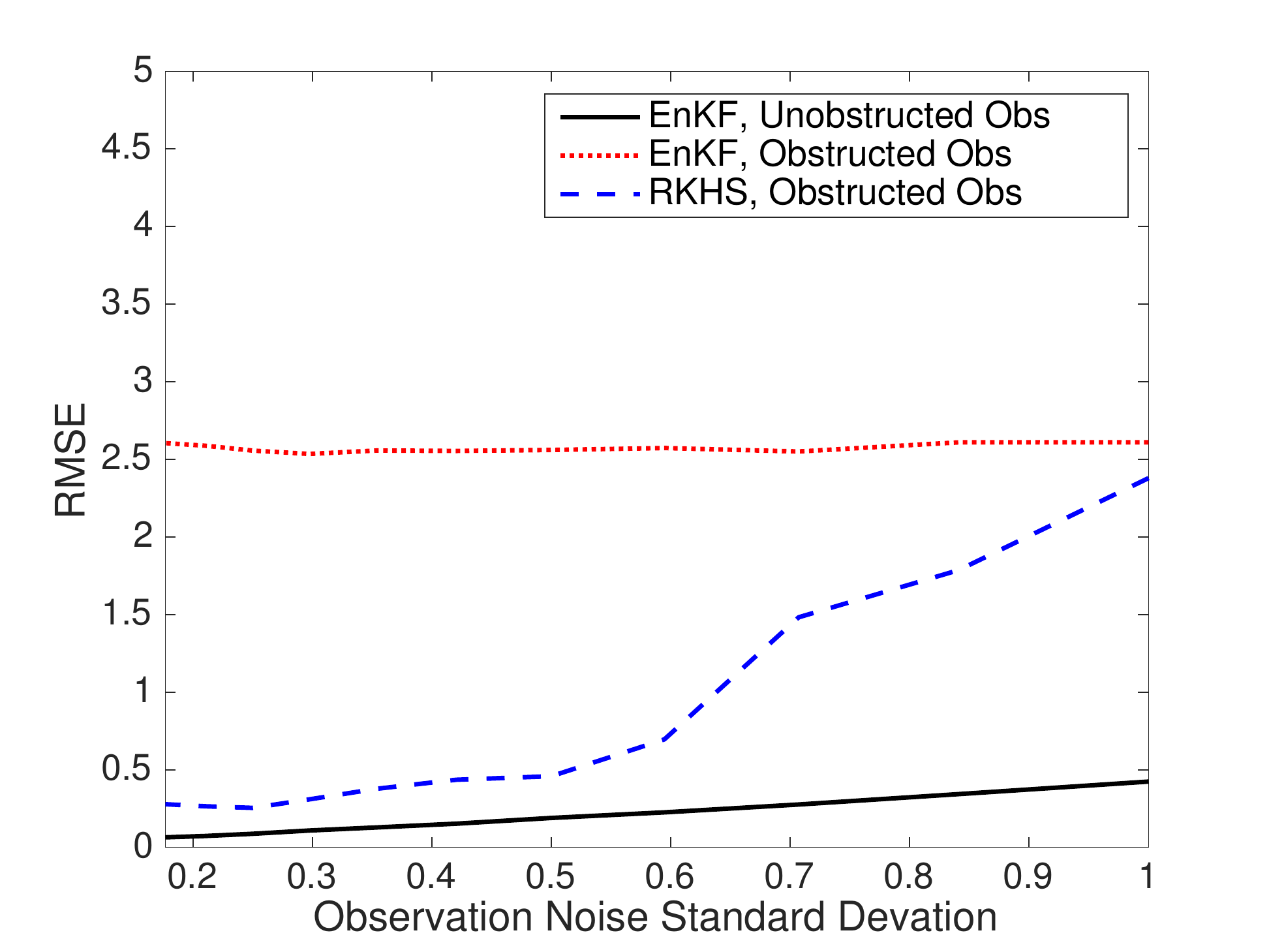}
\includegraphics[width=0.48\linewidth]{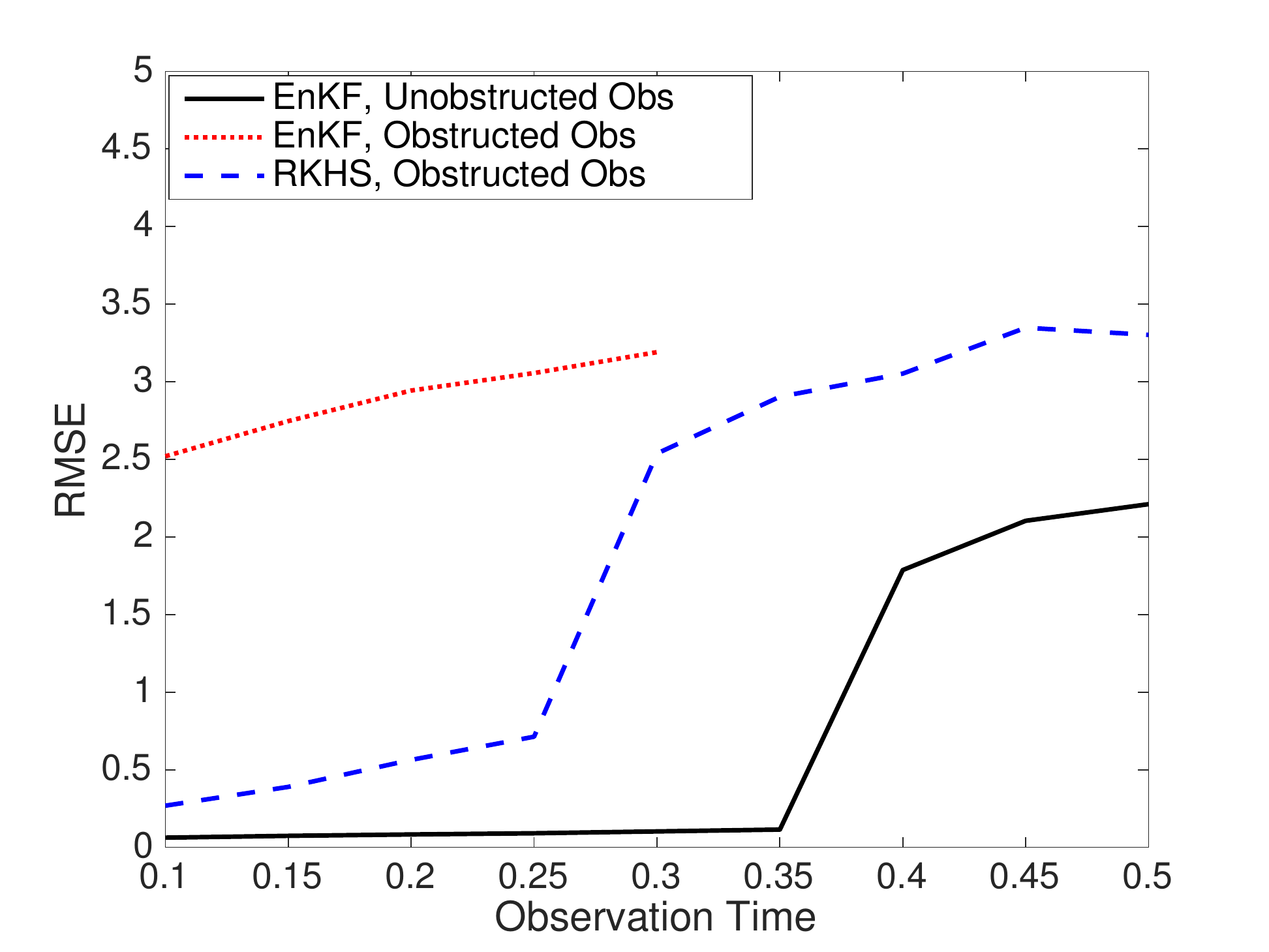}
\caption{\label{RandDT} Filtering improvement for cloudy observation using the RKHS bias correction (blue, dashed) compared to the standard EnKF with no model error correction (red, dotted) and compared to the EnKF which uses the perfect ``clear-sky" observations (black, solid). Top, Left: RMSE as a function of the measurement error standard deviation, $\sqrt{R^o}$ with a fixed observation time $\Delta t=0.1$.  Top, Right: RMSE as a function of observation time when $R^o=2^{-5}$. All filters in the top row used an $R$ parameter equal to the true observation noise covariance $R^o$ and fixing $Q=10^{-3}I$ (a simple additive inflation).  Bottom: Same as top row except that $Q$ and $R$ are estimated adaptively.  Notice that without the adaptive estimation scheme (top) the standard EnKF is often catastrophically divergent when filtering the obstructed observations.  The adaptive estimation scheme (bottom) improves the stability and performance of all the filters, but cannot overcome the obstructed observations in the EnKF (red, dotted) without the RKHS correction (blue, dashed).}
\end{figure}

\begin{figure}
\vspace{-10pt}
\centering
\includegraphics[width=0.48\linewidth]{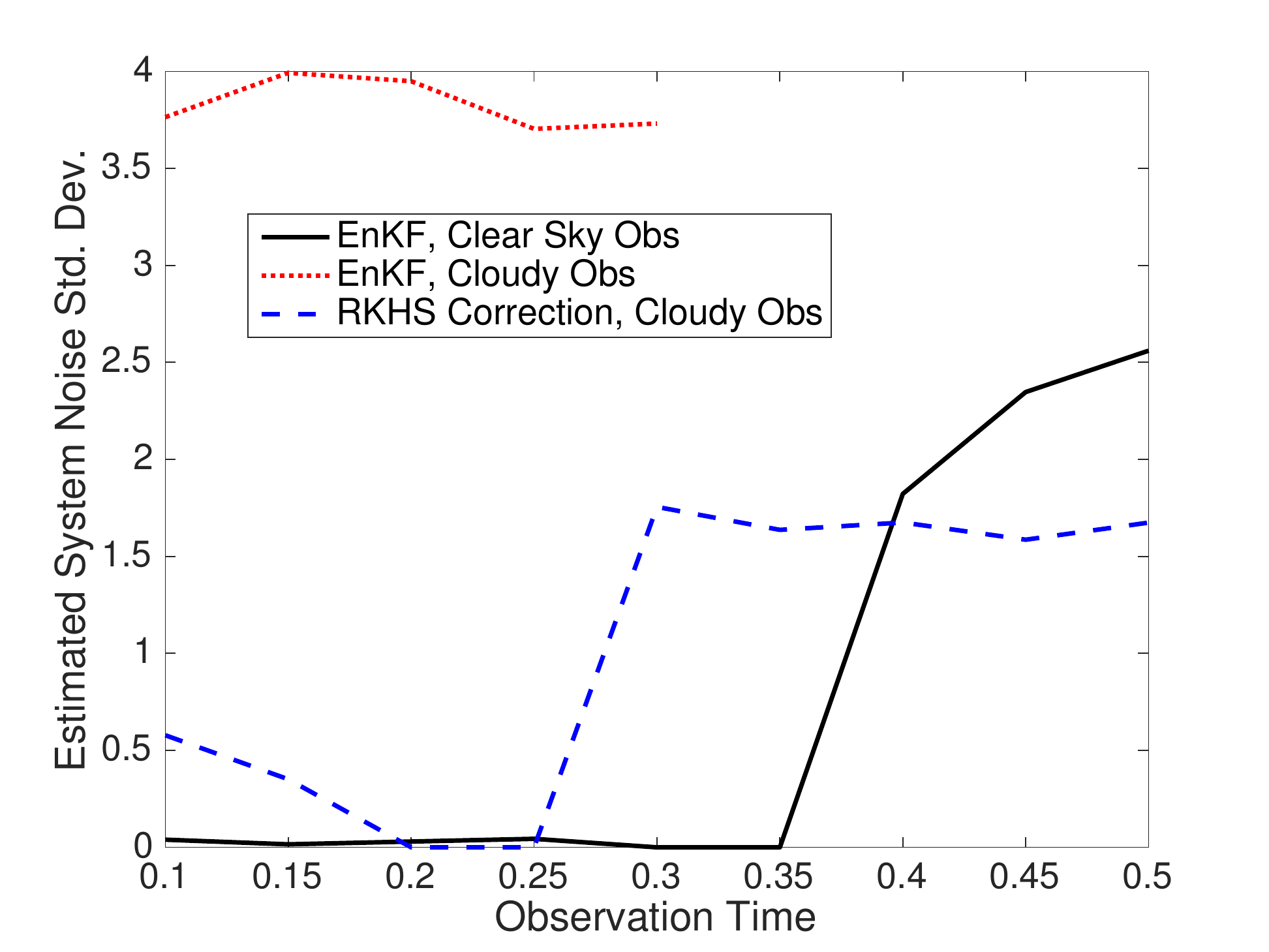}
\includegraphics[width=0.48\linewidth]{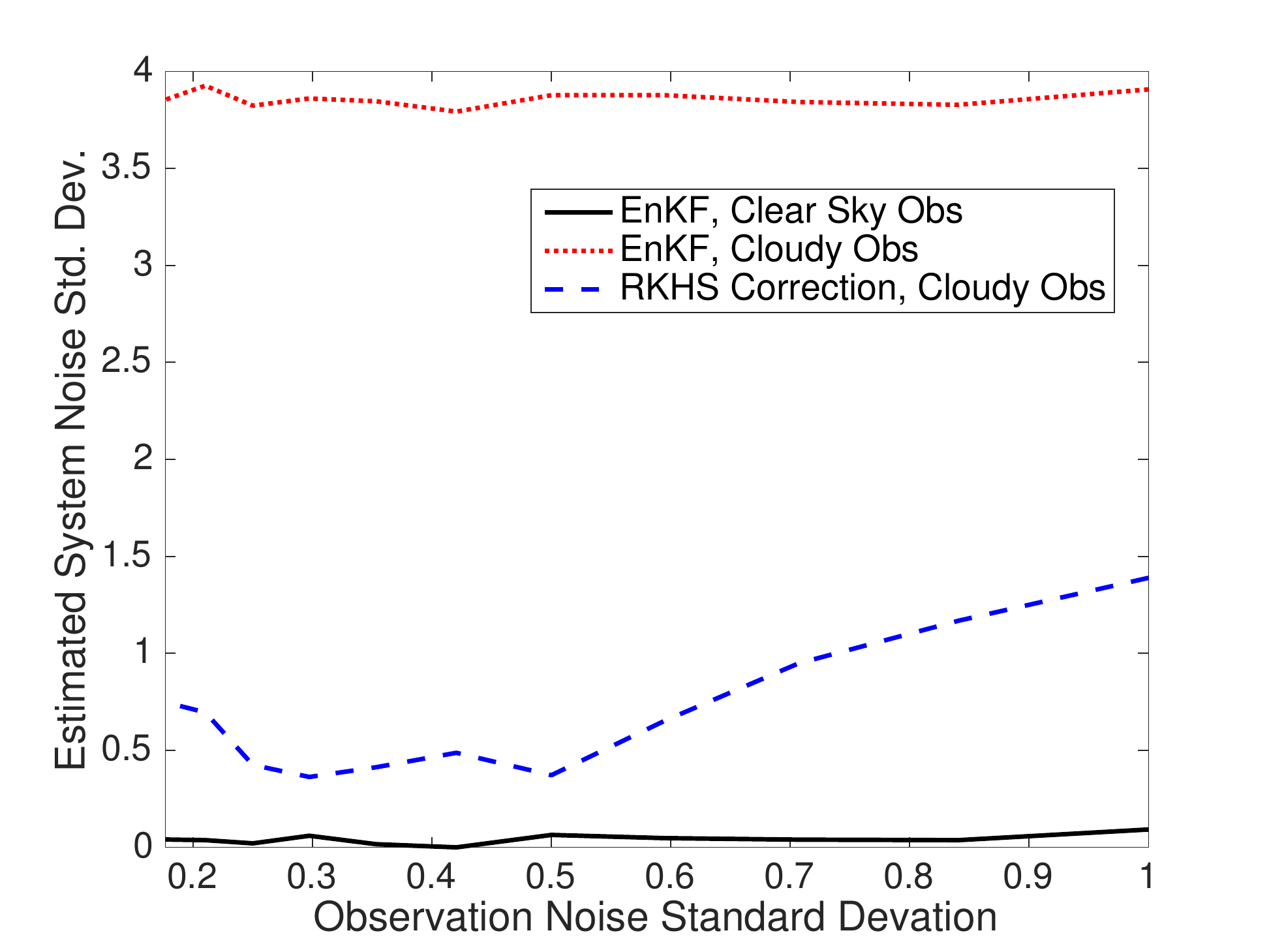}
\includegraphics[width=0.48\linewidth]{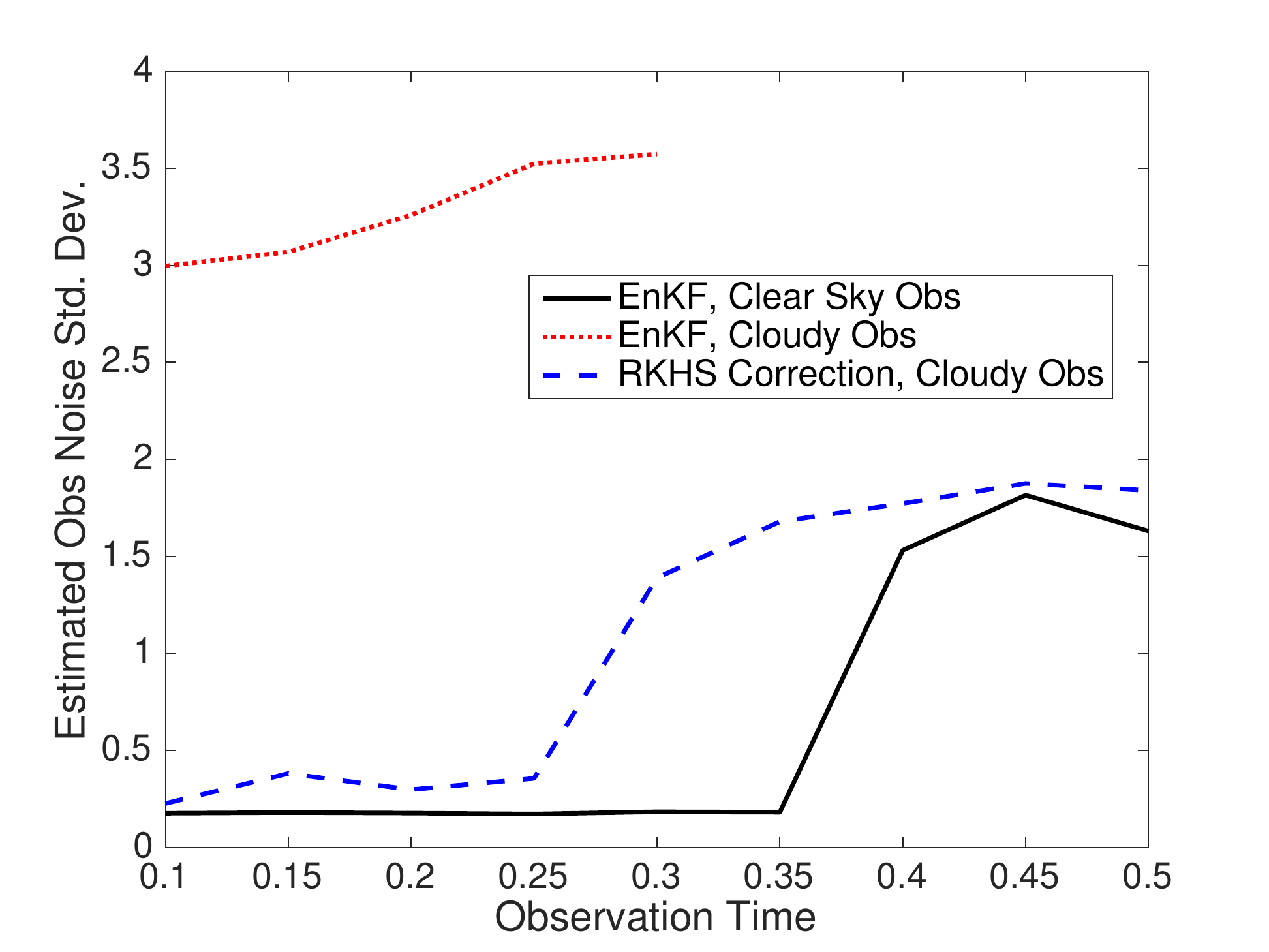}
\includegraphics[width=0.48\linewidth]{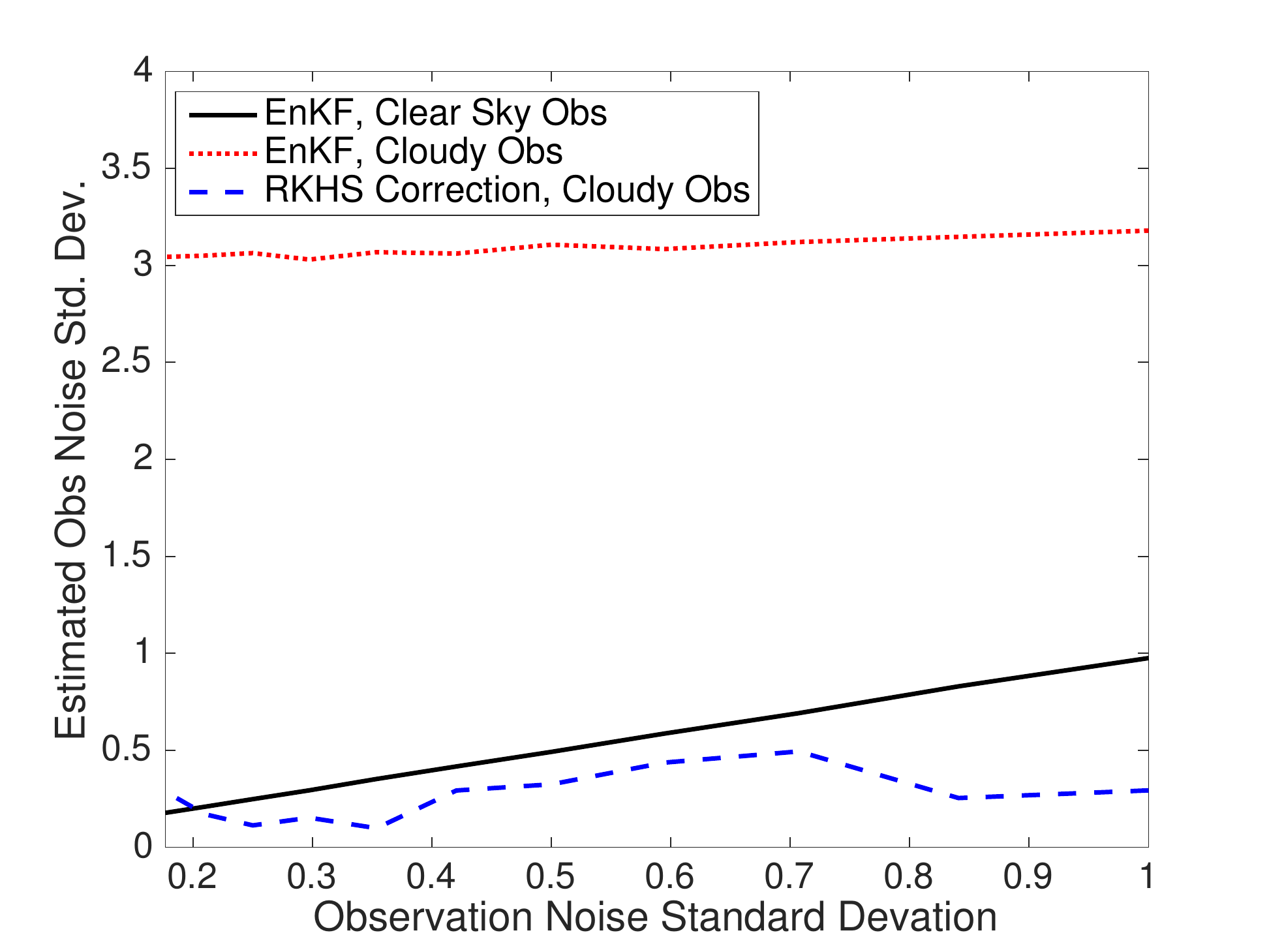}
\caption{\label{EstimatedQR} Estimated $Q$ and $R$ parameters for the various filters in Fig.~\ref{RandDT} using the adaptive estimation scheme \cite{bs:13}. }
\end{figure}

\section{Example 2: Assimilating cloudy ``satellite-like" measurements}

Here, we consider assimilating satellite measurements of an idealized stochastic cloud model for a single-column of organized tropical convection \citep{kbm:10}. The model is a system of four-dimensional ODEs that represents the time evolution of the first and second baroclinic anomaly potential temperatures, respectively, $\theta_1$ and $\theta_2$, an equivalent boundary layer anomaly potential temperature, $\theta_{eb}$, and a vertically averaged water vapor content, $q$, of a single-column tropical atmosphere. These ODEs are coupled to a stochastic birth-death process which governs the area fractions of three cloud types; congestus $f_c$, deep $f_d$, and stratiform clouds, $f_s$, over an unresolved horizontal domain corresponding to the column model. From $\theta_1$ and $\theta_2$, we can extrapolate the anomaly potential temperature at height $z$ (in km unit) using the following interpolation function \citep{kbm:10},
\BEA
T(z) = \theta_1 \sin(\frac{z\pi}{Z_T}) + 2 \theta_2 \sin(\frac{2z\pi}{Z_T}), \quad z\in[0,16],\label{temperature}
\EEA
where we model the height of the troposphere at the Tropic as $Z_T=16$ km.

In our experiment, we denote the state variables as $x=(\theta_1,\theta_2,\theta_{eb},q)$ and the cloud fractions as $f=(f_c,f_{d},f_s)$.  We will assume that the cloud top heights for both the deep and stratiform are the same, $z_d =12$ km, whereas that for the cumulus cloud is $z_c=3$ km. Based on these heights, we consider modeling the brightness temperature-like quantity at wavenumber-$\nu$ with the following two-cloud type radiative transfer model \citep{liou:02},
\begin{align}
h_\nu(x,f) &= (1-f_{d}-f_s)\Big[(1-f_c)\big( \theta_{eb} T_\nu(0)  \nonumber \\ &+ \int_0^{z_c} T(z) \frac{\partial T_\nu}{\partial z} (z)\,dz \big) \nonumber \\ &+ f_c T(z_c) T_\nu(z_c) + \int_{z_c}^{z_d} T(z) \frac{\partial T_\nu}{\partial z} (z)\,dz \Big]  \label{rtm} \\ &+ (f_d+f_s)T(z_d) T_\nu(z_d) + \int_{z_d}^\infty T(z) \frac{\partial T_\nu}{\partial z} (z)\,dz,\nonumber
\end{align}
where $T_\nu(z)$ denotes the transmission from height $z$ to the satellite, which is assumed to be a decaying function of height and also depends on $q$. We specify the wavenumbers $\nu$ such that the weighting functions, $\frac{\partial T_\nu}{\partial z} (z)$, are maximized at heights $z_{max}=1, 2, \ldots, 16$ km. In Fig.~\ref{weight}, we show the weighting functions (black solid) with the minimum humidity associated with $z_{max}=2, 5, 8, 10$ km (black dashes). We also include the weighting functions corresponding to the maximum humidity value (red solid). This is to show that depending on the value of $q$, the shape of weighting functions will vary between these two weights. (see Appendix C for the detailed construction of these weighting functions). 

\begin{figure}
\centering
\includegraphics[width=.6\linewidth]{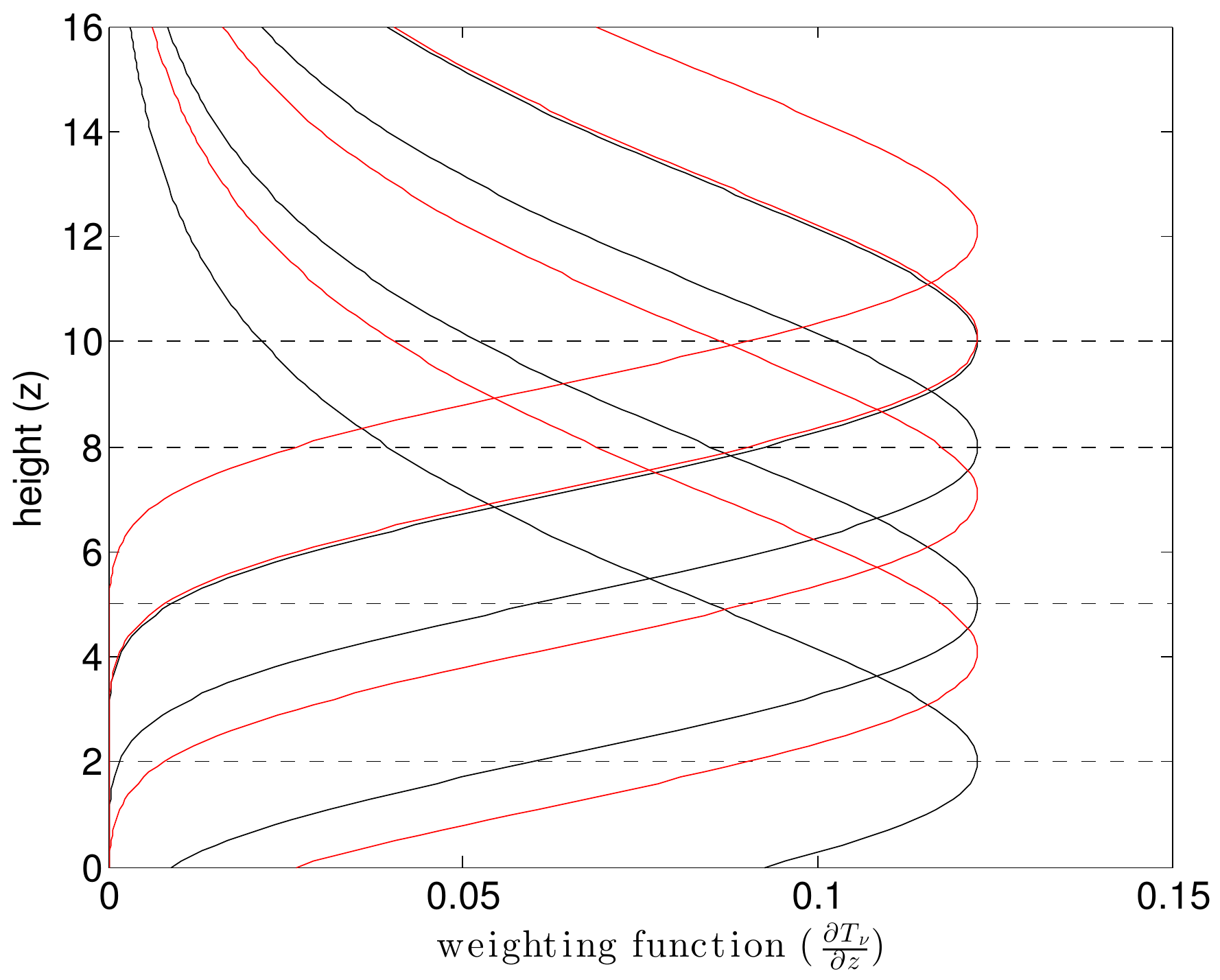}
\caption{The weighting functions for the minimum value of $q$ (black solid) which is used as a reference to define the absorption rate corresponding to $z_{max}= 2, 5, 8, 10$ km. We also show the weighting functions for the maximum value of $q$ corresponding to these maximum heights (red).} 
\label{weight}
\end{figure}

Essentially, the terms inside the square bracket in \eqref{rtm} denote the contribution below the deep and stratiform clouds, the first term on the third row denotes the contribution from the deep cloud top height and the last term denotes the contribution of the free atmosphere above the deep cloud. By setting $f=0$, we obtain the clear-sky observations. Consider implementing the primary filter with an incorrect observation model $\tilde{h}_\nu(x) = h_\nu(x,0)$, which means that we assume the observed brightness temperatures are from clear-sky measurements. In Fig.~\ref{cloudR1}, we show the scatter plot of the observations which are defined as $y_\nu = h_\nu(x,f)+\eta_\nu$, as functions of the model error, $b_\nu=h_\nu(x,f)-\tilde{h}_\nu(x)$. In this example the measurement errors are specified as a percentage of the variance of the measured variables, so $\eta_\nu\sim\mathcal{N}(0,R^o)$ with $R^o=10^{-3} \times \textup{var}(h_\nu(x,f))$ where $\textup{var}(h_\nu(x,f))$ is the variance of the observation.  We will also consider the robustness to the measurement error $R^o$ below.

\begin{figure}
\centering
\includegraphics[width=.95\linewidth]{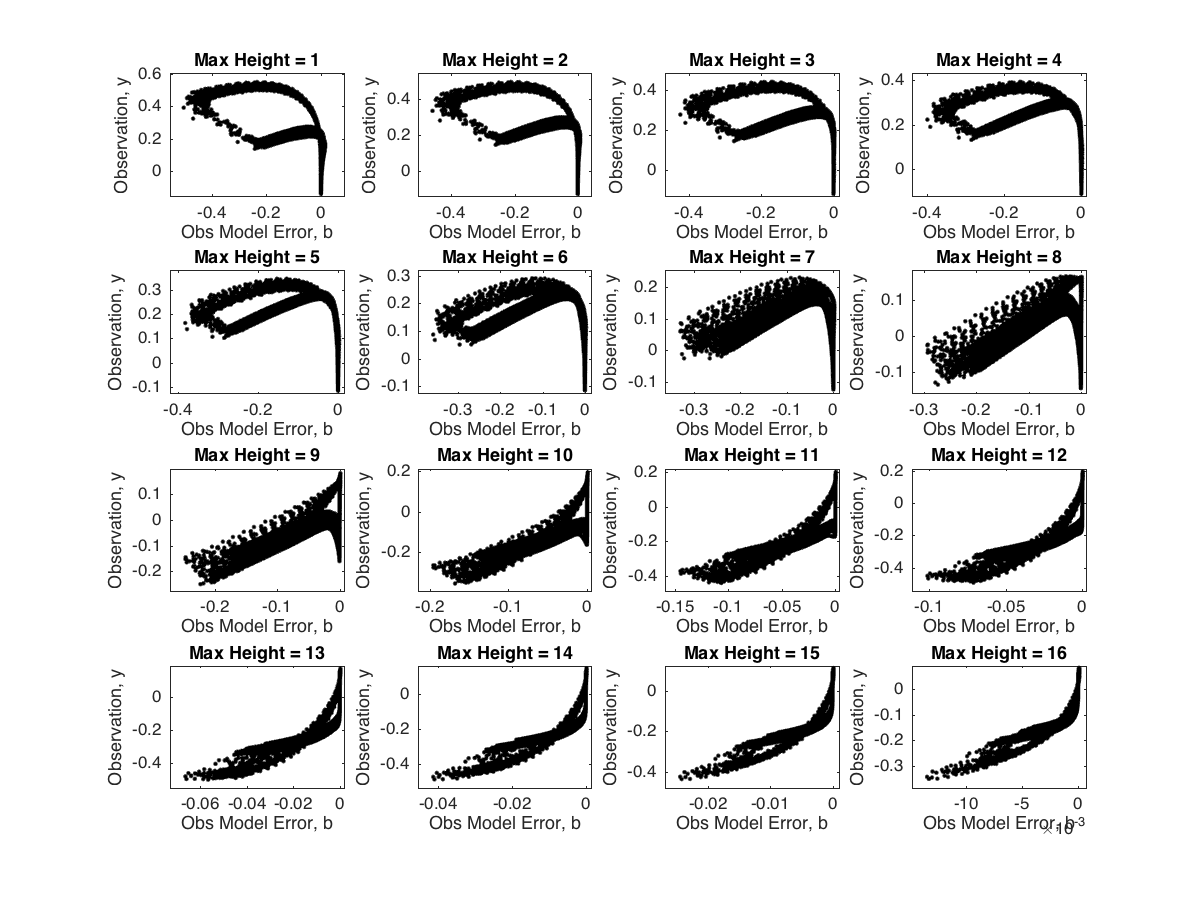}
\caption{\label{cloudR1} Scatter plots of the cloudy observations, $y_\ell$, versus the observation model errors, $b_\ell$, for the observation at 16 frequencies-$\nu$ corresponding to weighting functions with maximum heights at $z_{max}=1,2,...,16$, respectively.}
\end{figure}

We observe at each observation time, $\Delta t= 0.0035$ and use an ensemble size of $K=14$. In Fig.~\ref{secondaryfilter}, we show examples of the secondary filter at two instances for observations at wavenumber-$\nu$ corresponding to weighting function with $z_{max}=6$ km. The two observations are relatively close to each other, $y_i=0.184$ and $y_i=0.195$, but the former has very small observation error while the latter has a relatively large observation error (see the blue dots in these two figures). Here, the nonparametric likelihood functions are trained with $N=5000$ and $M=400$. To correct the observation model error, we ran the secondary filter in \eqref{Bayes2} with the Gaussian prior with time-dependent variances $\sigma_{b_i}^2$ as defined in \eqref{sigmab} (grey curves). In this example, we could not use the time-independent variance, $\sigma_b^2$, from the training data $b_\ell$ since the resulting prior is too broad and was uninformative. However, since $\sigma_{b_i}^2$ is very small (see e.g. Fig.~\ref{cloudR1}), it is possible that the support of the prior will be almost entirely be between the two modes of a bimodal likelihood function, which will cause the posterior normalization factor $Z$ to be very small.  Thus, when $Z$ is found to be less than a threshold $Z^*$, we do not apply the secondary filter, since in these cases the likelihood function does not give enough information to inform the secondary filter. When the prior is on the tail of the likelihood function, we do not apply this thresholding step. In our numerical implementation, we found our results to be robust to a large range of thresholds $Z^* \in [10^{-30},10^{-2}]$. 

\begin{figure}
\centering
\includegraphics[width=0.48\linewidth]{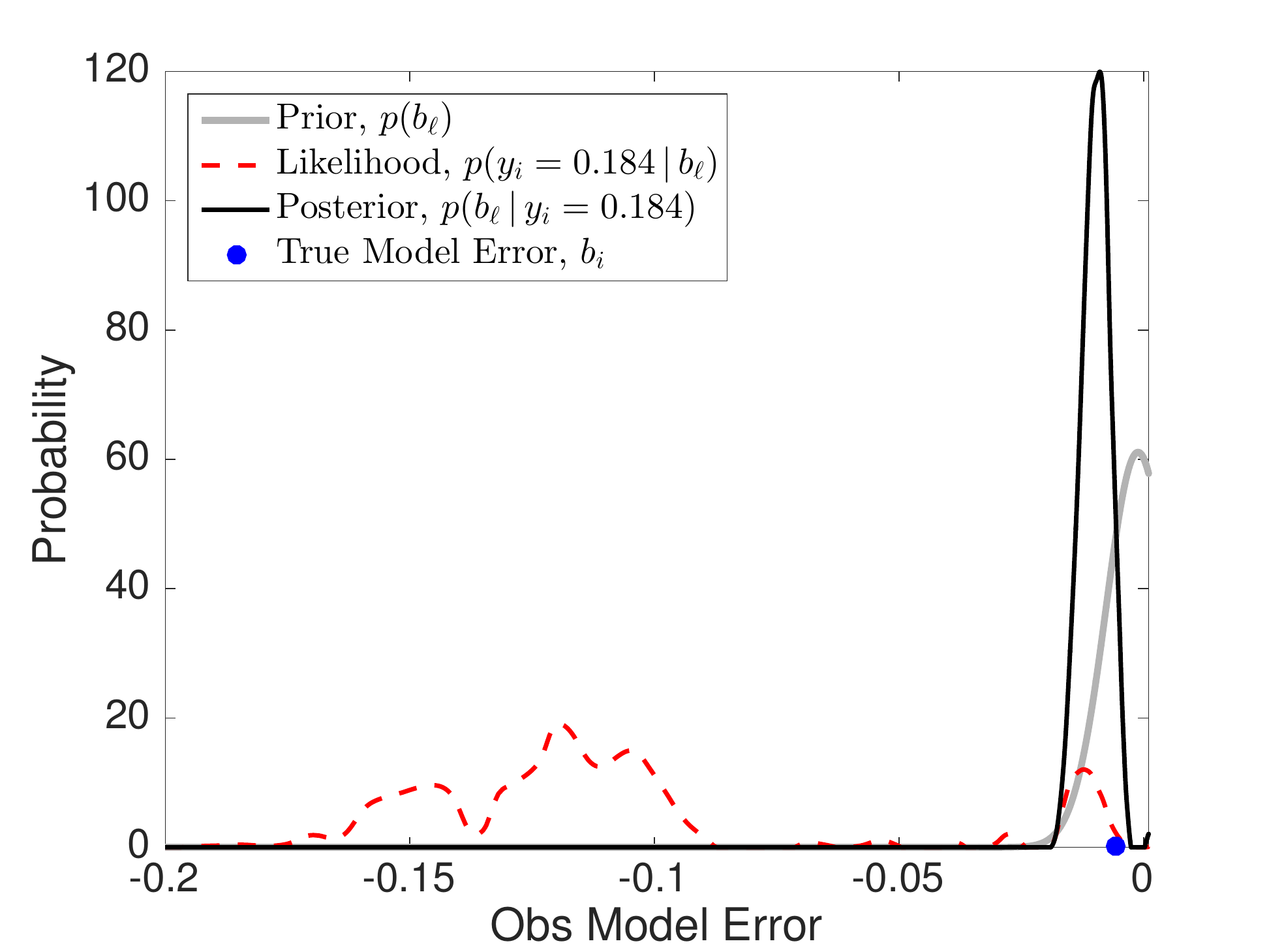}
\includegraphics[width=0.48\linewidth]{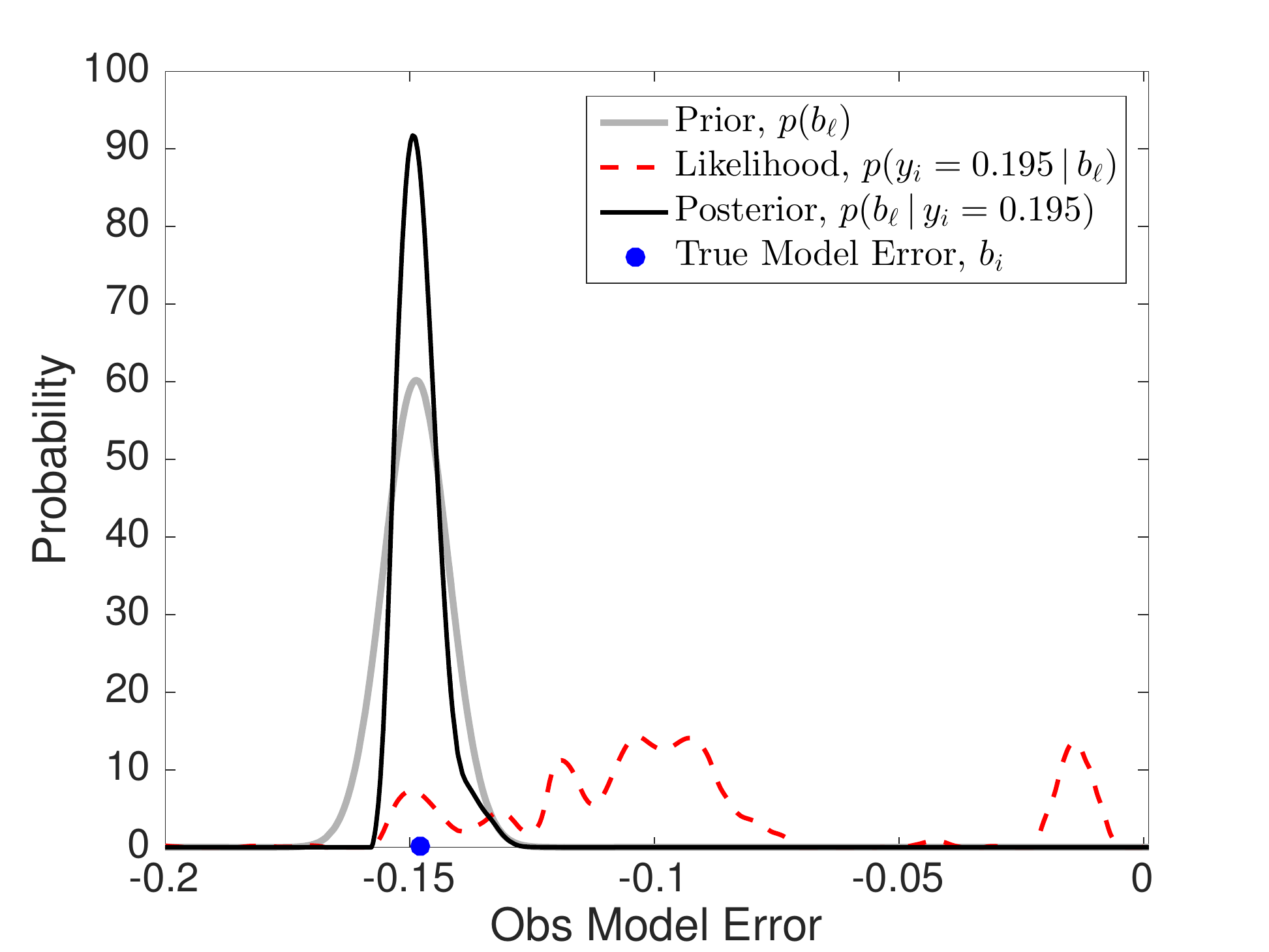}
\caption{\label{secondaryfilter} Bimodal likelihood and Bayesian update from the secondary filter at two random instances, one with small observation model error (left) and the other with large error (right). These two plots are for the observations at wavenumber-$\nu$ corresponding to weighting function with $z_{max}=6$ km.}
\end{figure}

In Fig.~\ref{cloud1}, we show the filter estimates as timeseries of all seven model variables $\theta_1, \theta_2, \theta_{eb}, q, f_c, f_d, f_s$ when the measurement error is $R^o=10^{-3}\times \textup{var}(h_\nu(x,f))$ or 0.1\% of the observation variance. For diagnostic, we also include EnKF assimilating the same cloudy observations with the true observation function. Notice that the RKHS is as accurate as the EnKF implemented with the true observation function. On the other hand, the EnKF based on the incorrect observation model does not produce accurate estimates, despite using an inflated observation covariance matrix in the filter (a standard method for overcoming model error). In this example we do not use any adaptive covariance estimation, all filters use an additive covariance inflation $Q=10^{-3} \times I_{7\times 7}$, and an inflated $R=5R^o$.   

In Fig.~\ref{cloudRobust} we compare the filter Relative Mean Square Error (RMSE) as a function of the measurement error variance, $R^o$, ranging from 0.1\% to 2\% of the observation variance. Here, the Relative MSE is defined as the ratio between the MSE and the variance of the corresponding state variable, where MSE is averaged over 2500 filter steps. As we saw in the L96 example above, when the measurement error is small the RKHS performance approaches that of the EnKF given the true observation function. For large measurement error, the performance of all of the filters degrades significantly, even the EnKF given the true observation function. In Fig.~\ref{cloud2} we show the filter estimates as timeseries for all seven model variables at the highest measurement error which is 2\% of the observation variance. In this case, notice that even the EnKF given the true observation function produces inaccurate filter estimates for some variables (with large errors in $q, f_s, f_d, f_c$). This illustrates the difficulty of filtering these highly nonlinear observations given by the RTM in the presence of large measurement error.  For small measurement errors, the RKHS correction is again effective at overcoming the severe bias introduced by the observation model error.

\begin{figure}
\centering
\includegraphics[width=0.48\linewidth]{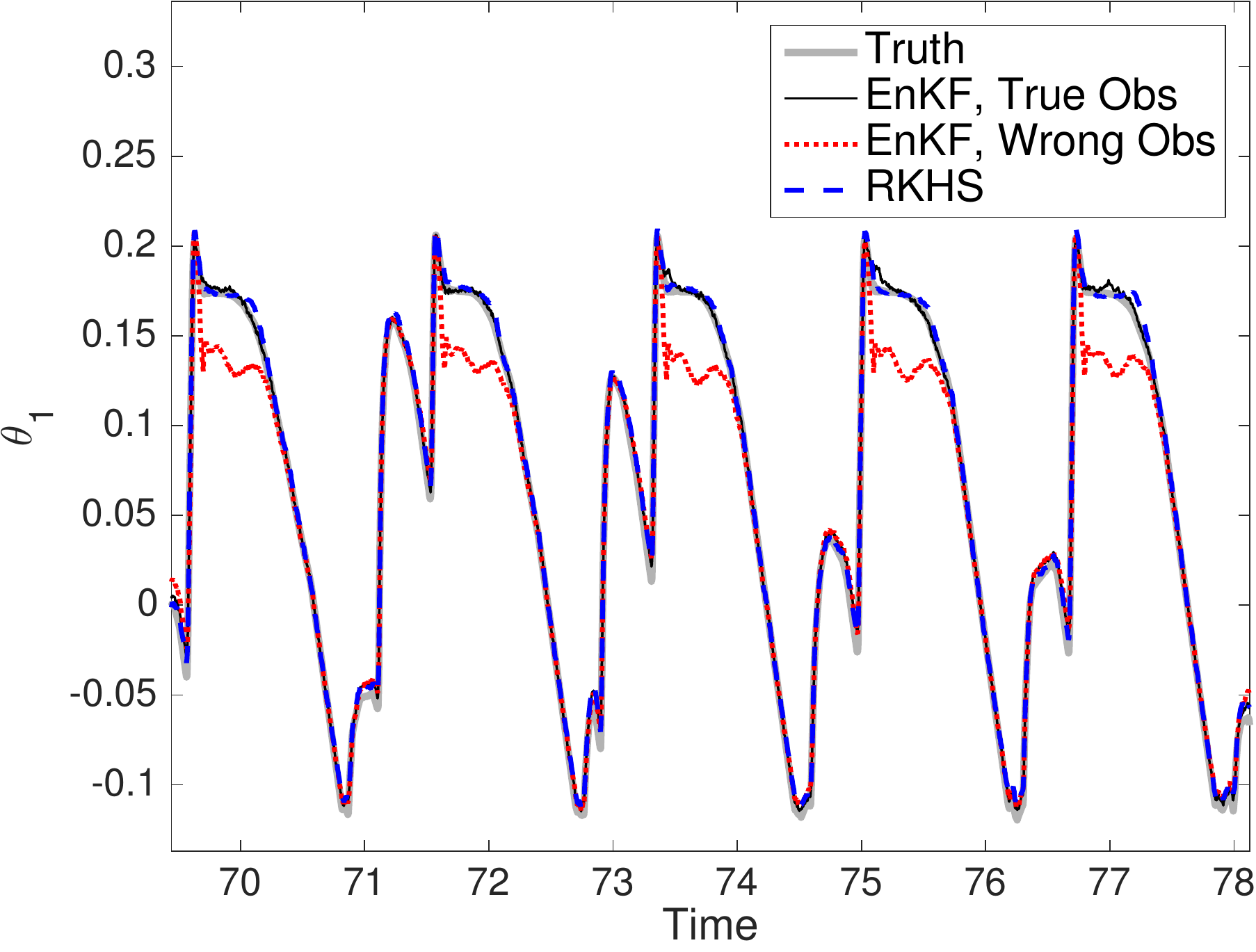}
\includegraphics[width=0.48\linewidth]{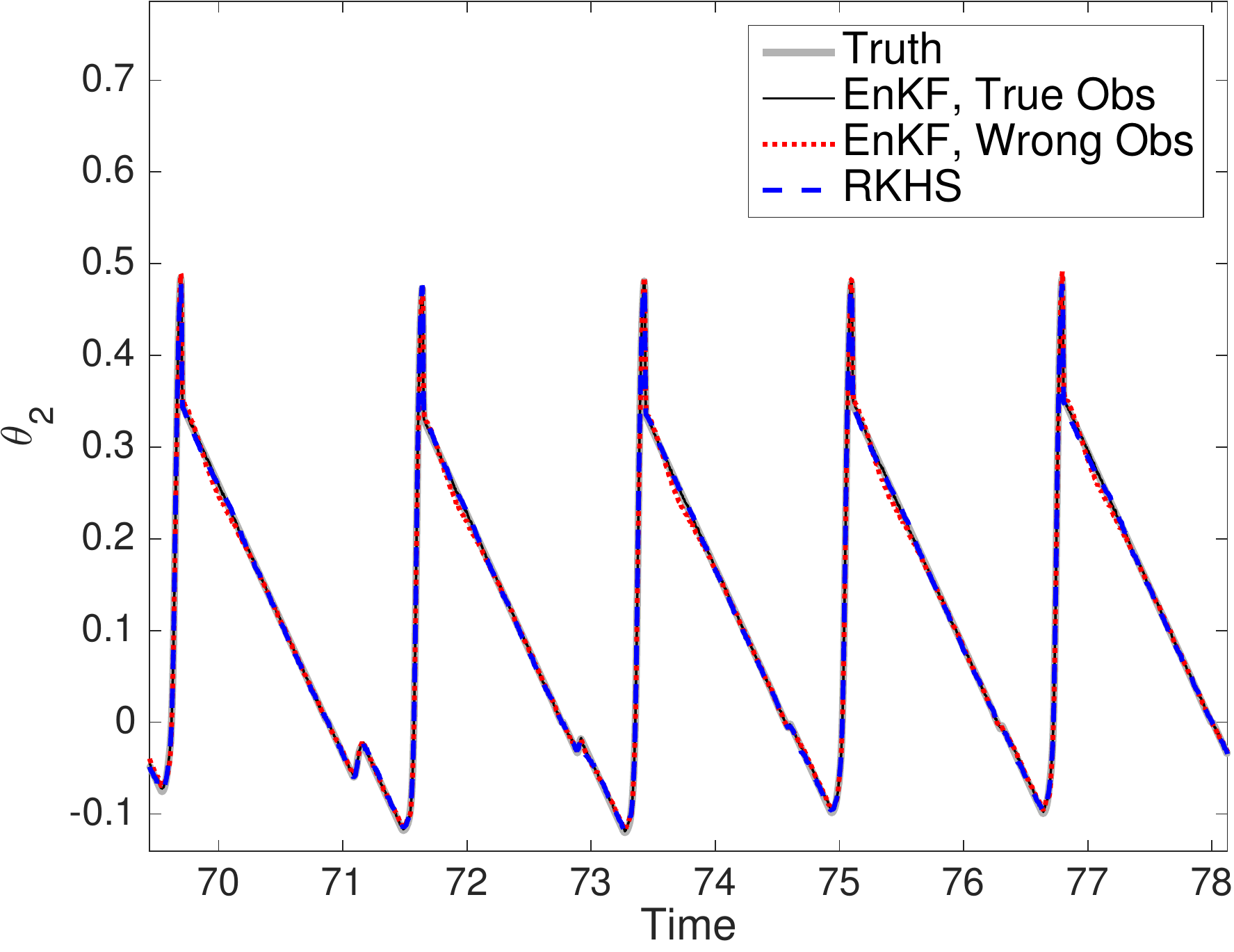}
\includegraphics[width=0.48\linewidth]{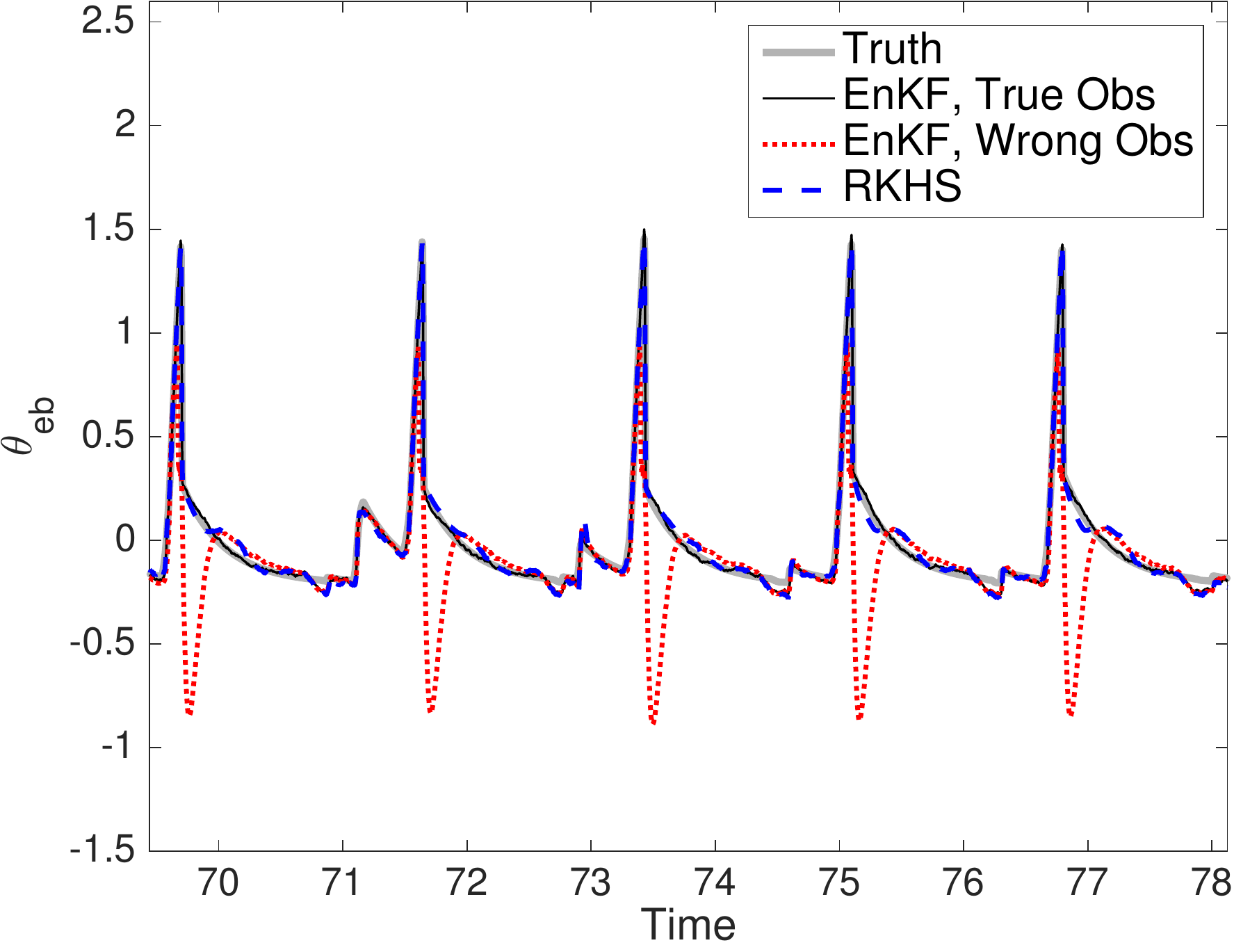}
\includegraphics[width=0.48\linewidth]{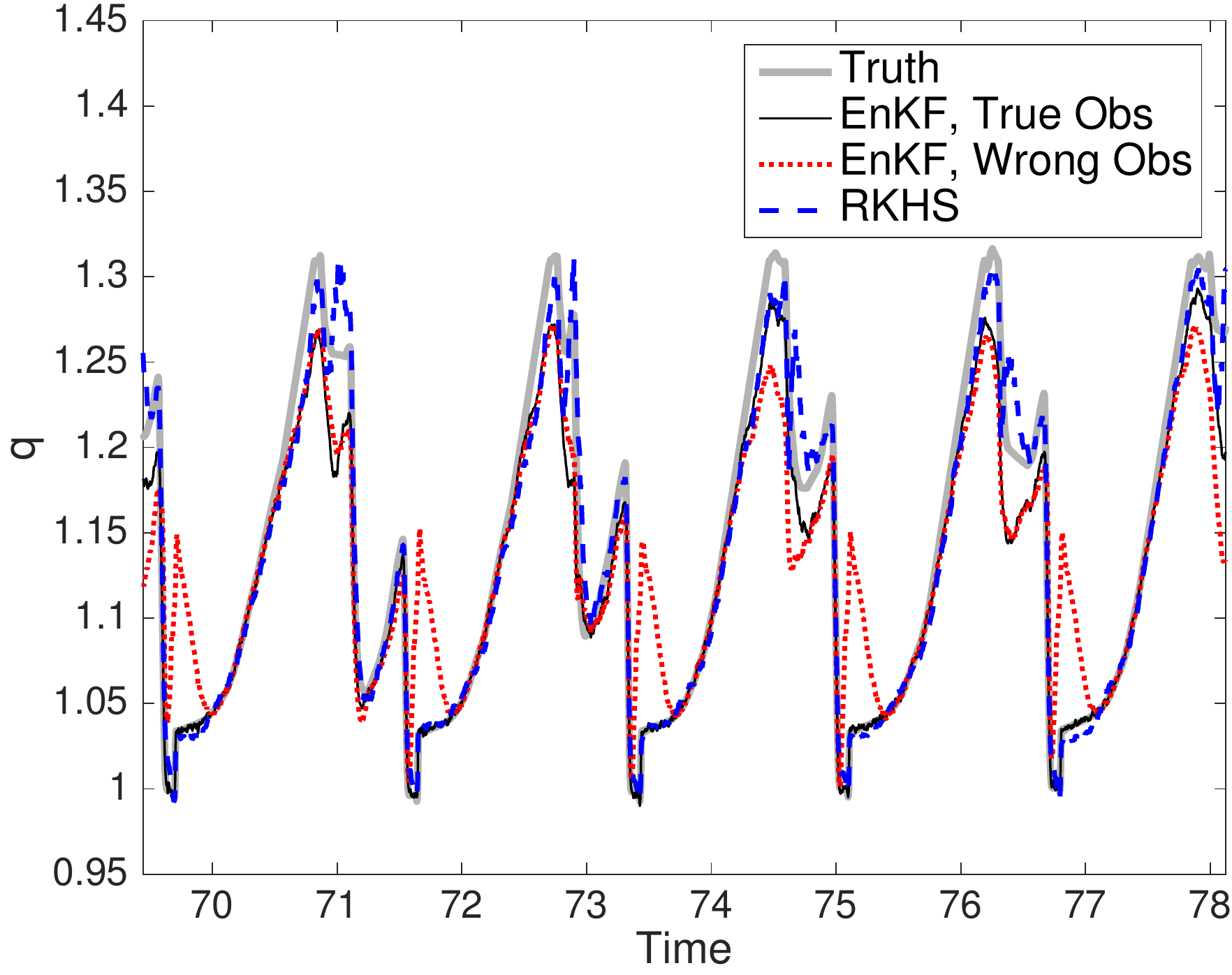}
\includegraphics[width=0.48\linewidth]{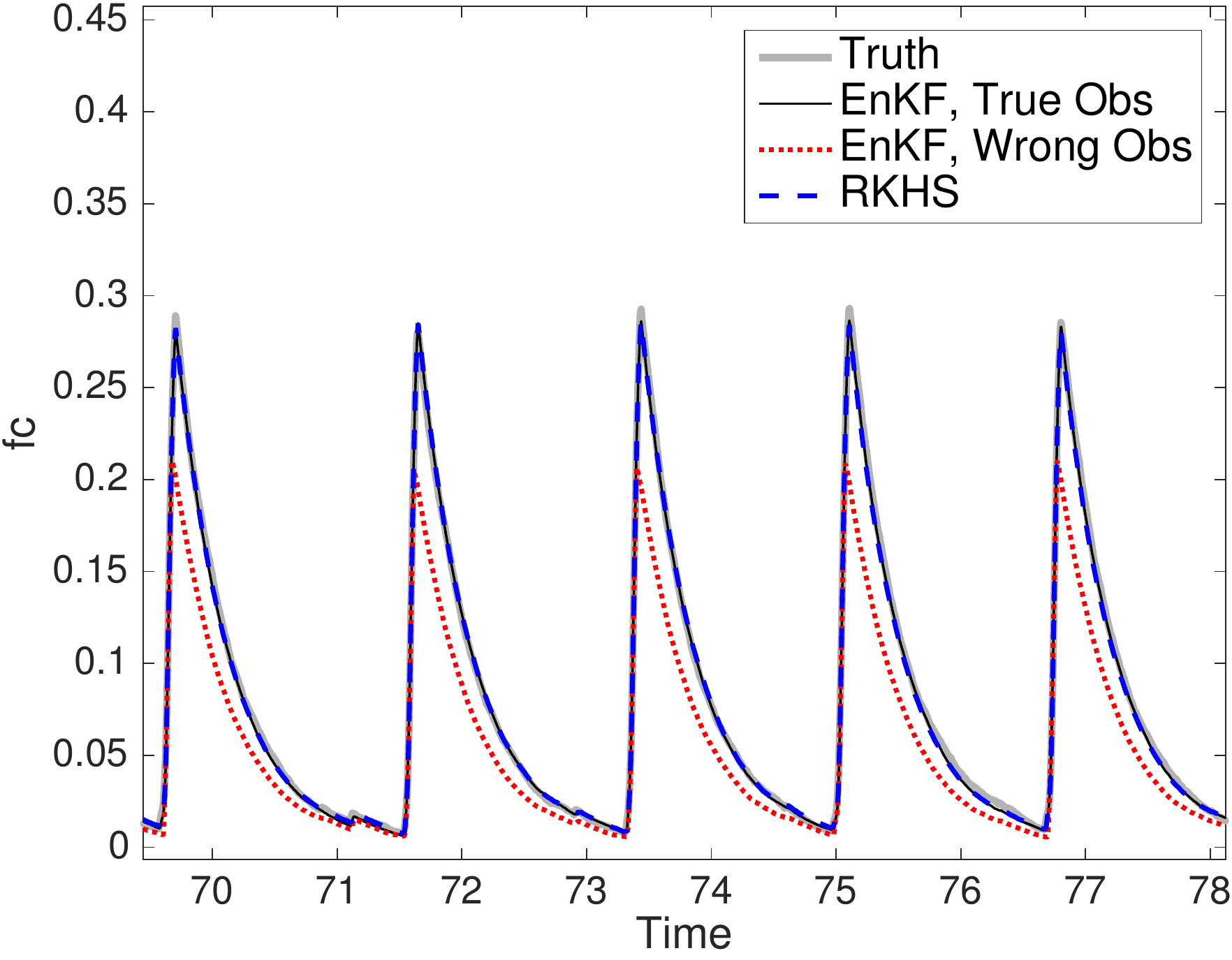}
\includegraphics[width=0.48\linewidth]{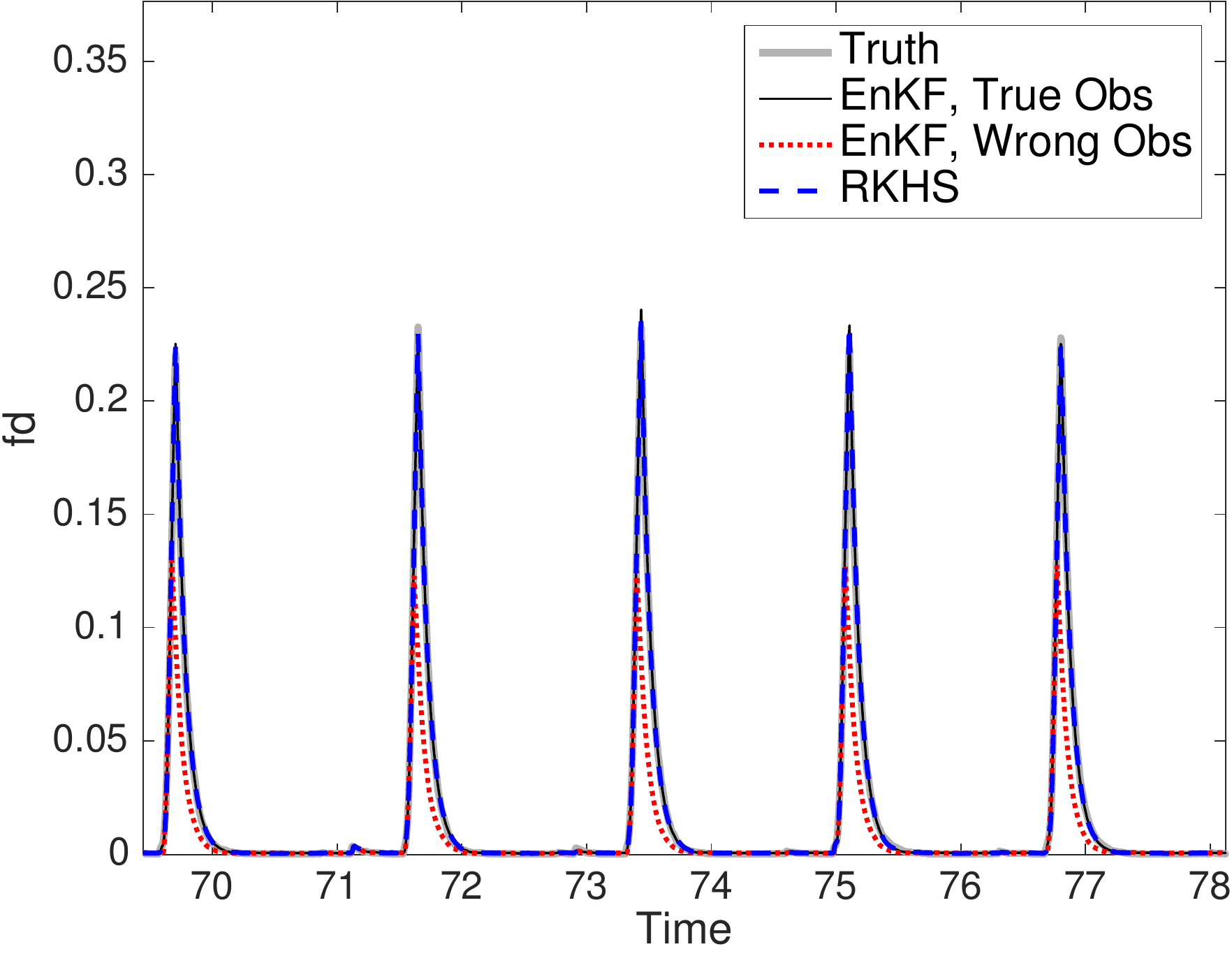}
\includegraphics[width=0.48\linewidth]{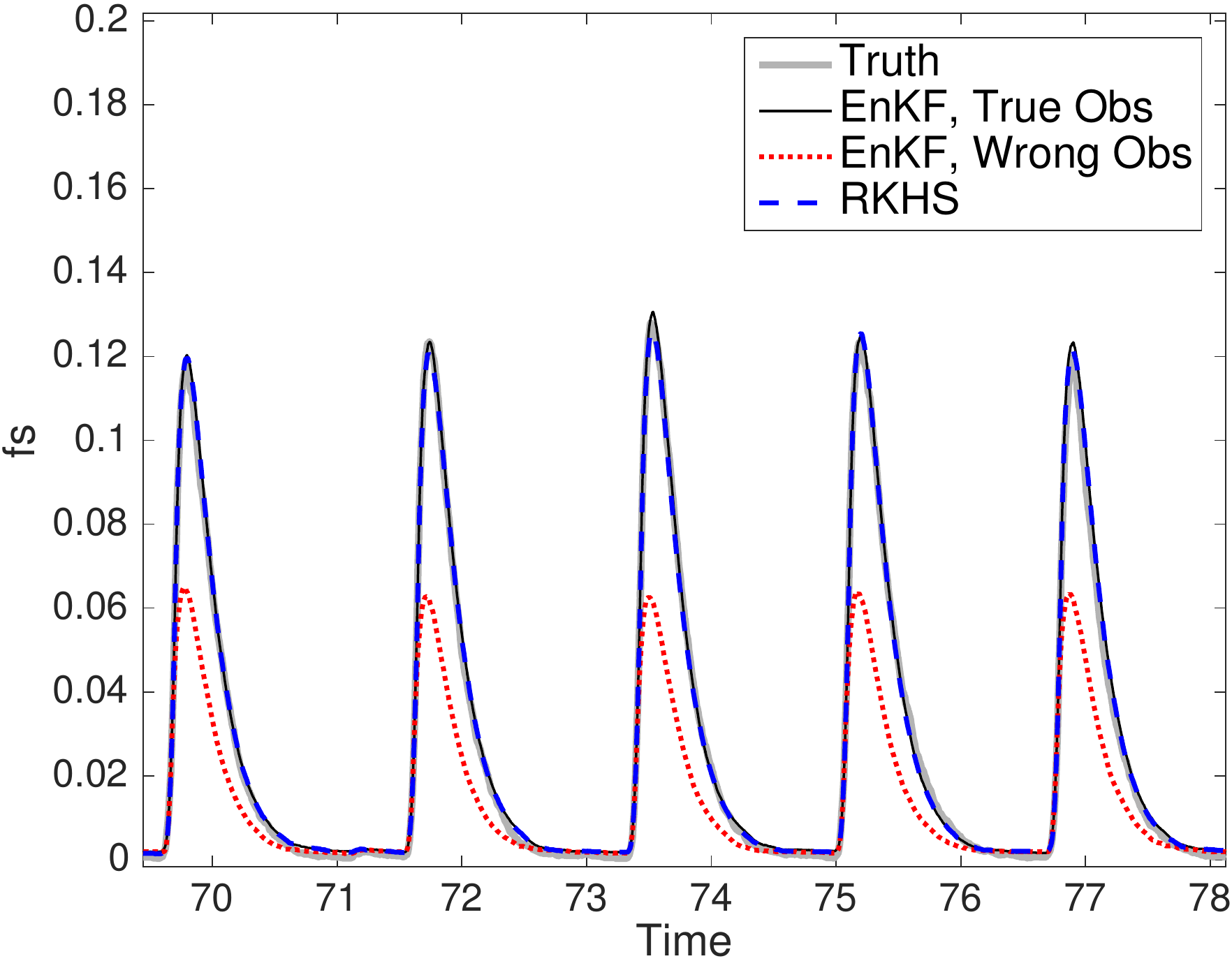}
\caption{\label{cloud1} Time series of filter estimates compared to the true state variables for all 7 model variables in the presence of small measurement error $R^o=10^{-3}\times \textup{var}(h_\nu(x,f))$ which is 0.1\% of the observation variance.}
\end{figure}

\begin{figure}
\centering
\includegraphics[width=0.32\linewidth]{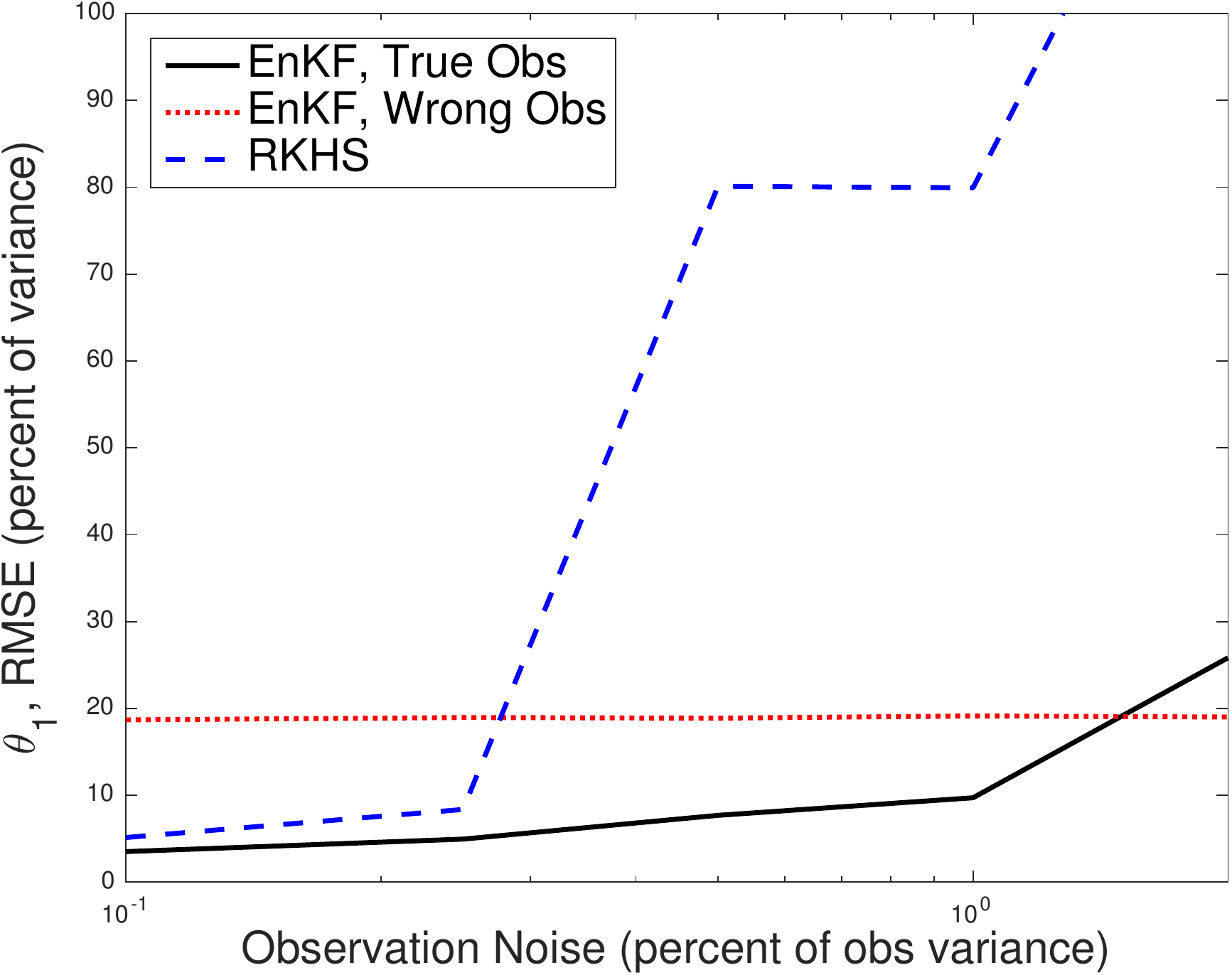}
\includegraphics[width=0.32\linewidth]{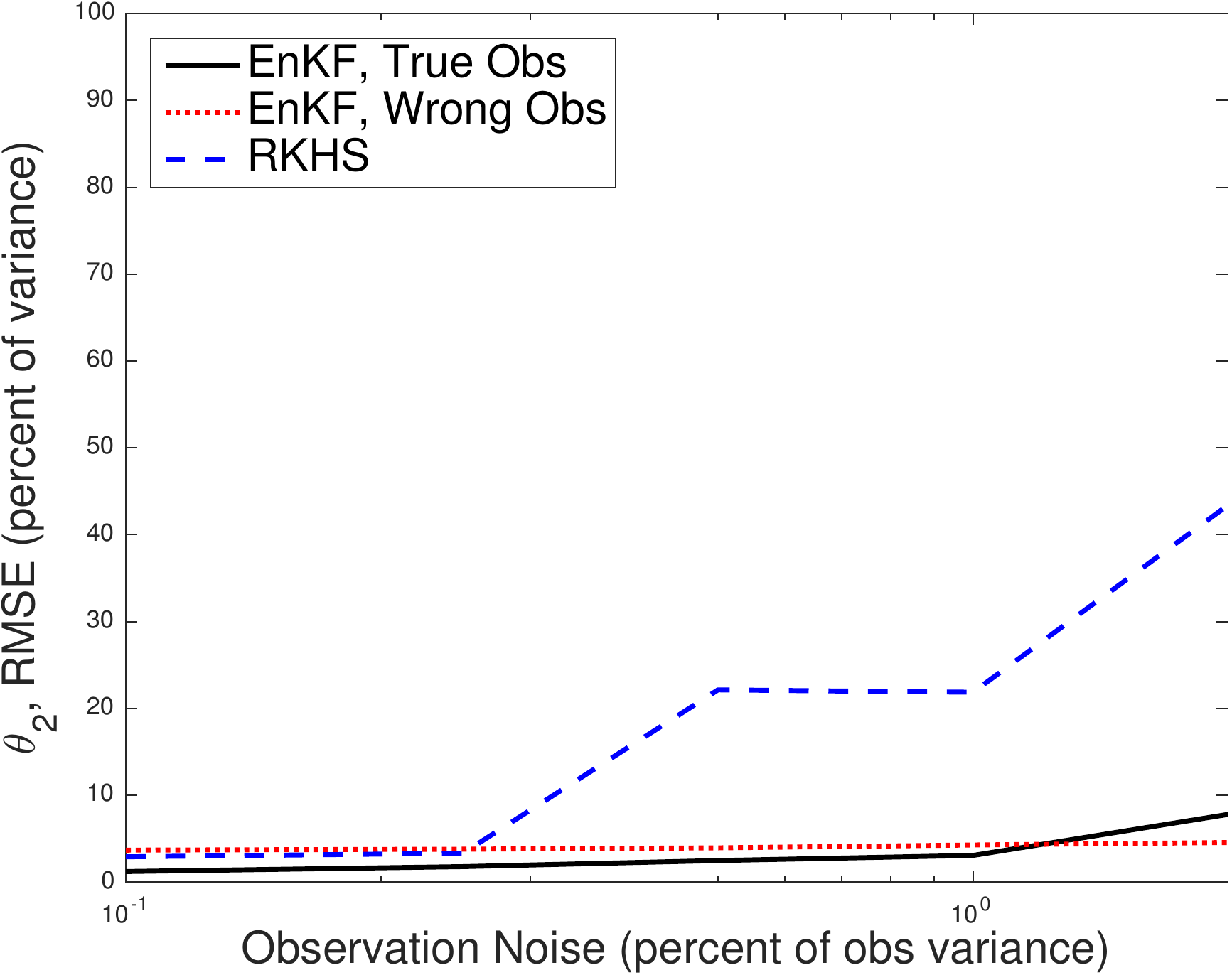}
\includegraphics[width=0.32\linewidth]{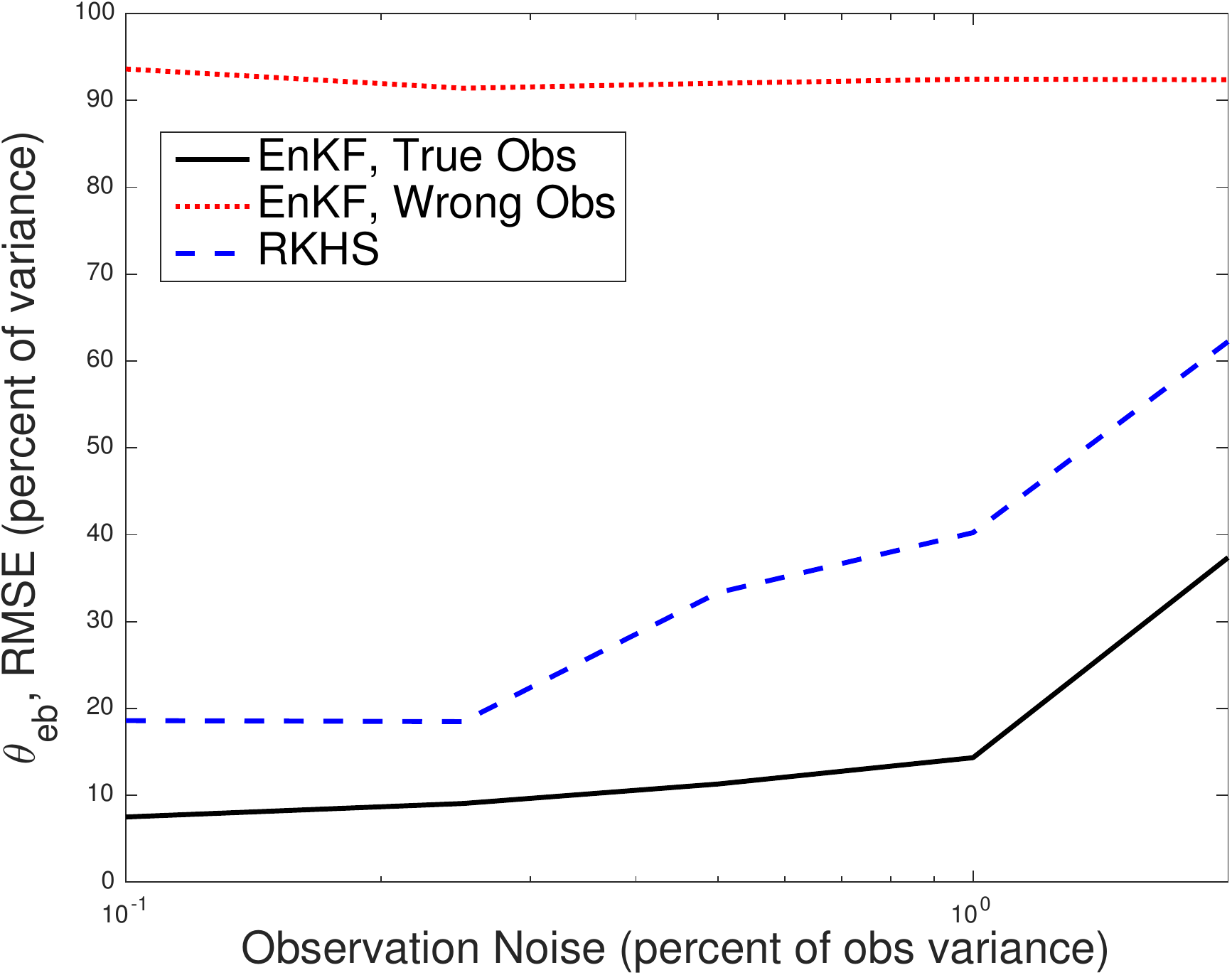}
\includegraphics[width=0.32\linewidth]{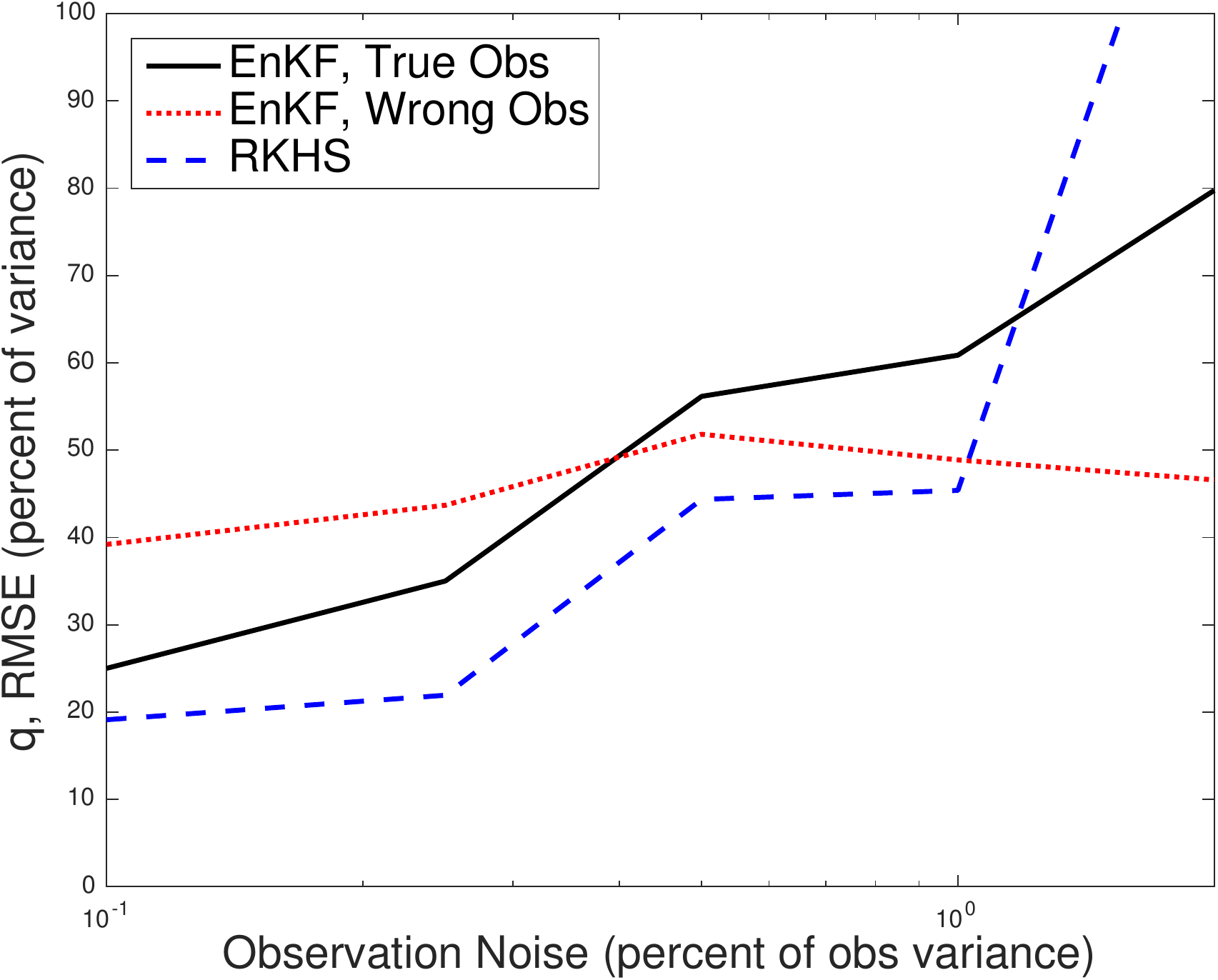}
\includegraphics[width=0.32\linewidth]{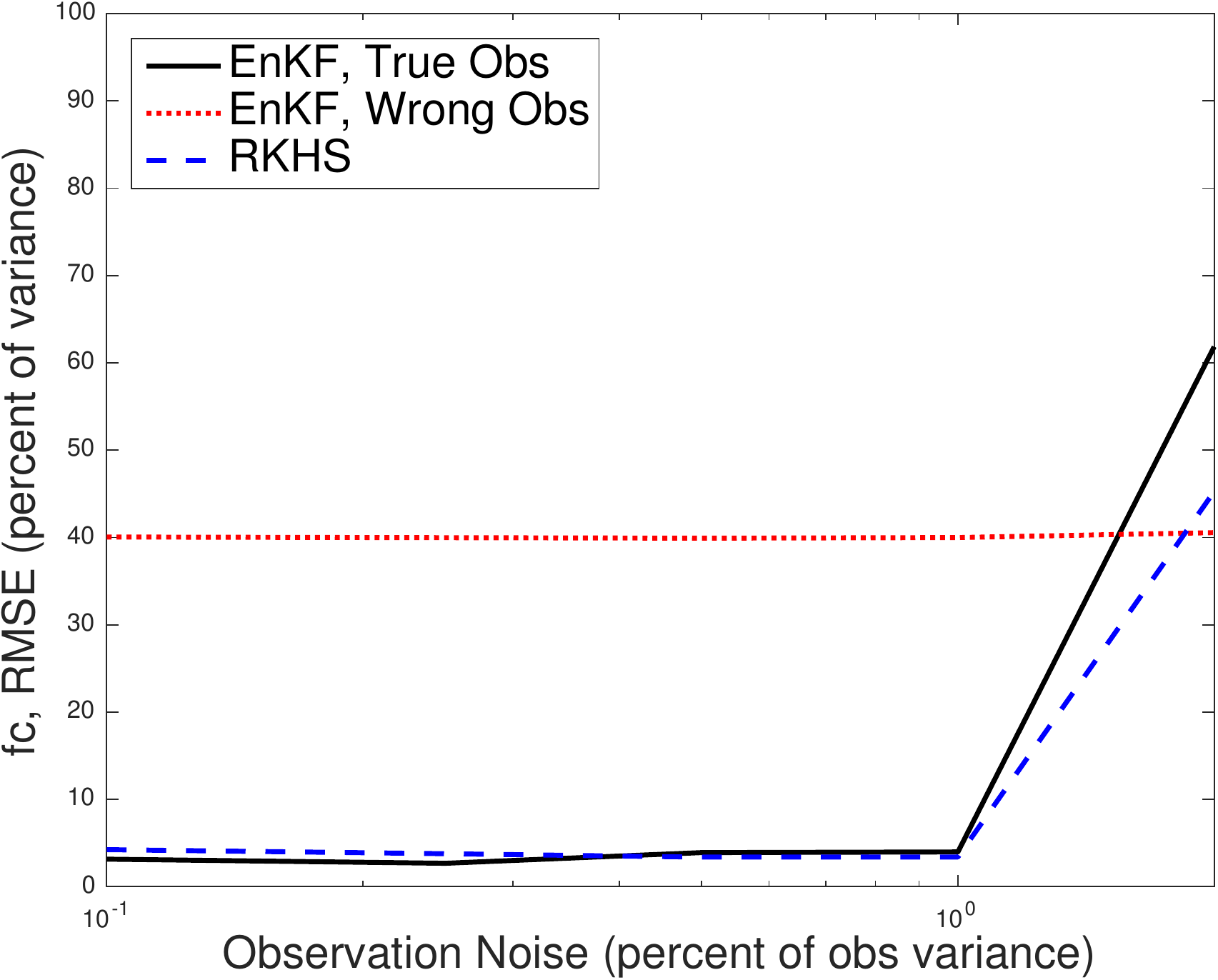}
\includegraphics[width=0.32\linewidth]{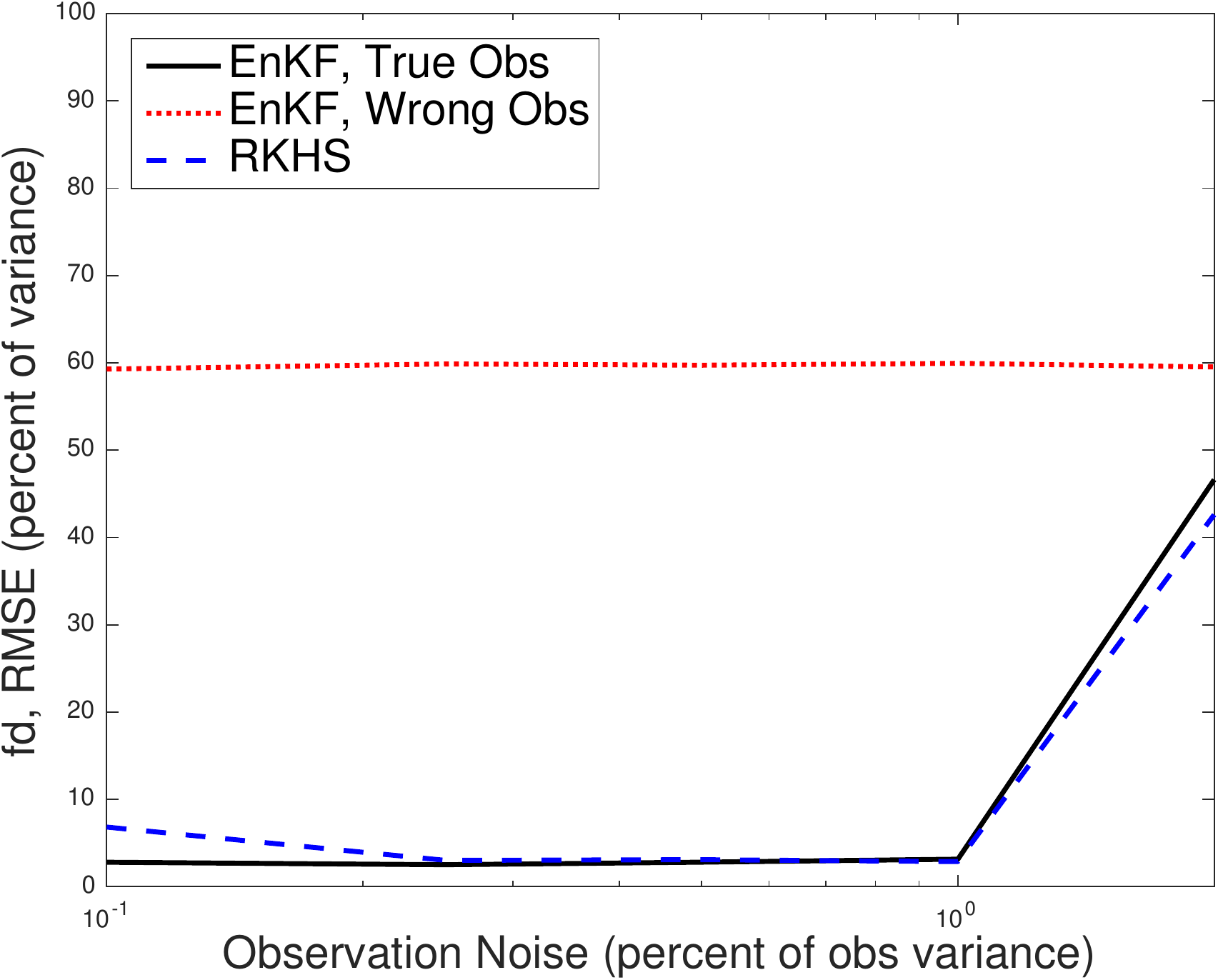}
\includegraphics[width=0.32\linewidth]{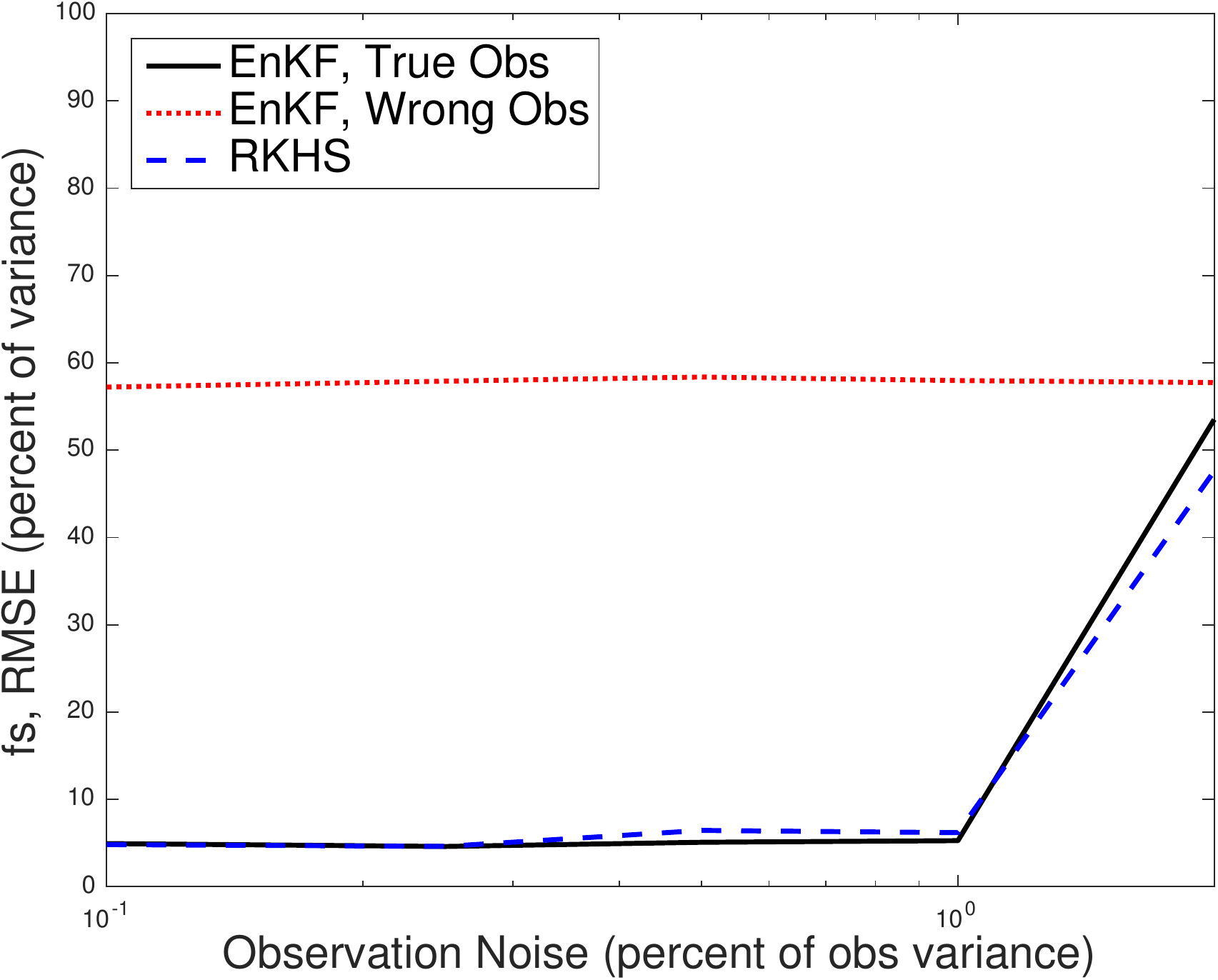}
\caption{\label{cloudRobust} Robustness of the various filters to increasing measurement error $R^o$ from 0.1\% up to 2\% of of the observation variance. Performance is measured by Relative Mean Squared Error to the variance of the corresponding variable.  Notice that for large measurement error the performance of all the filters decreases signficantly.}
\end{figure}

\begin{figure}
\centering
\includegraphics[width=0.48\linewidth]{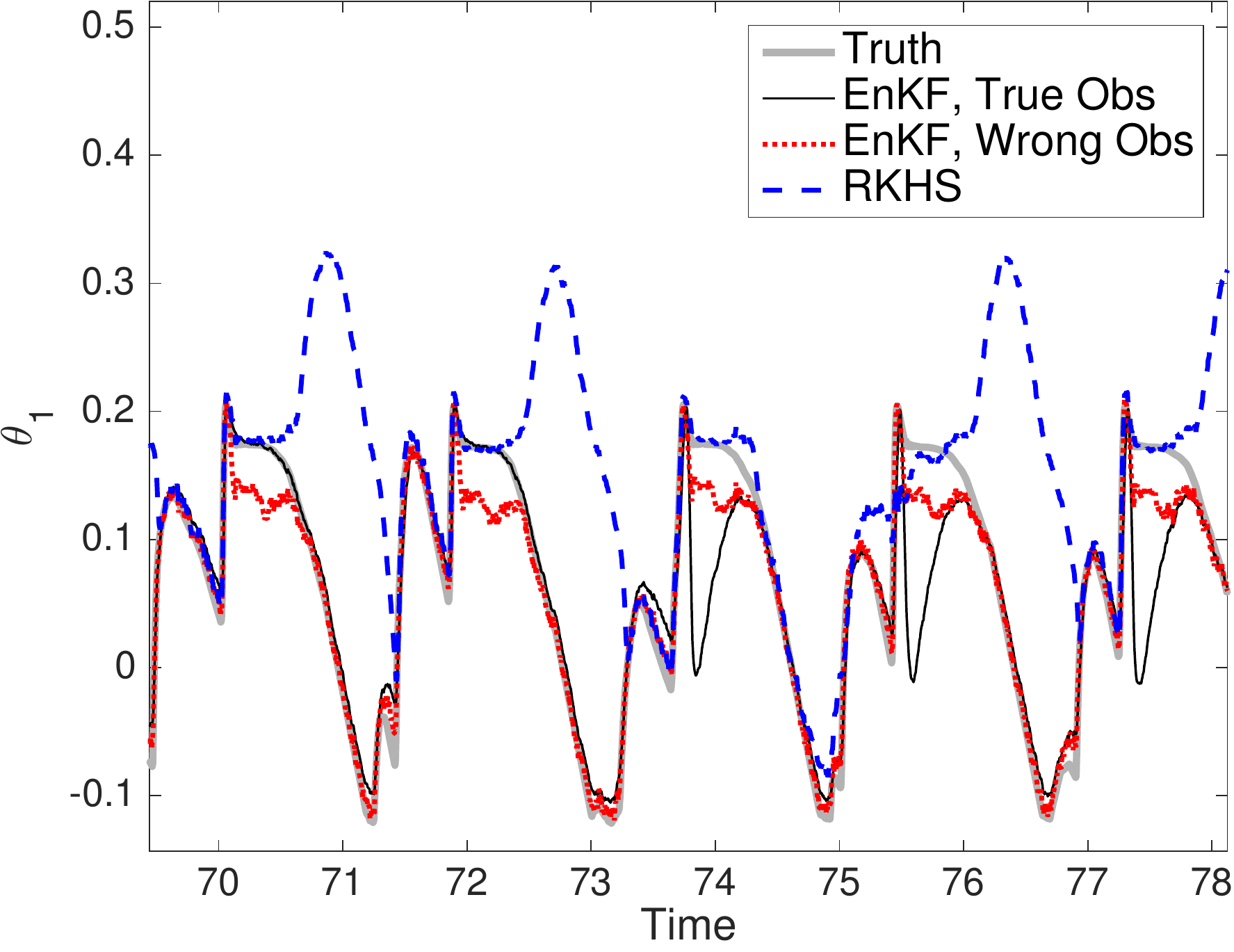}
\includegraphics[width=0.48\linewidth]{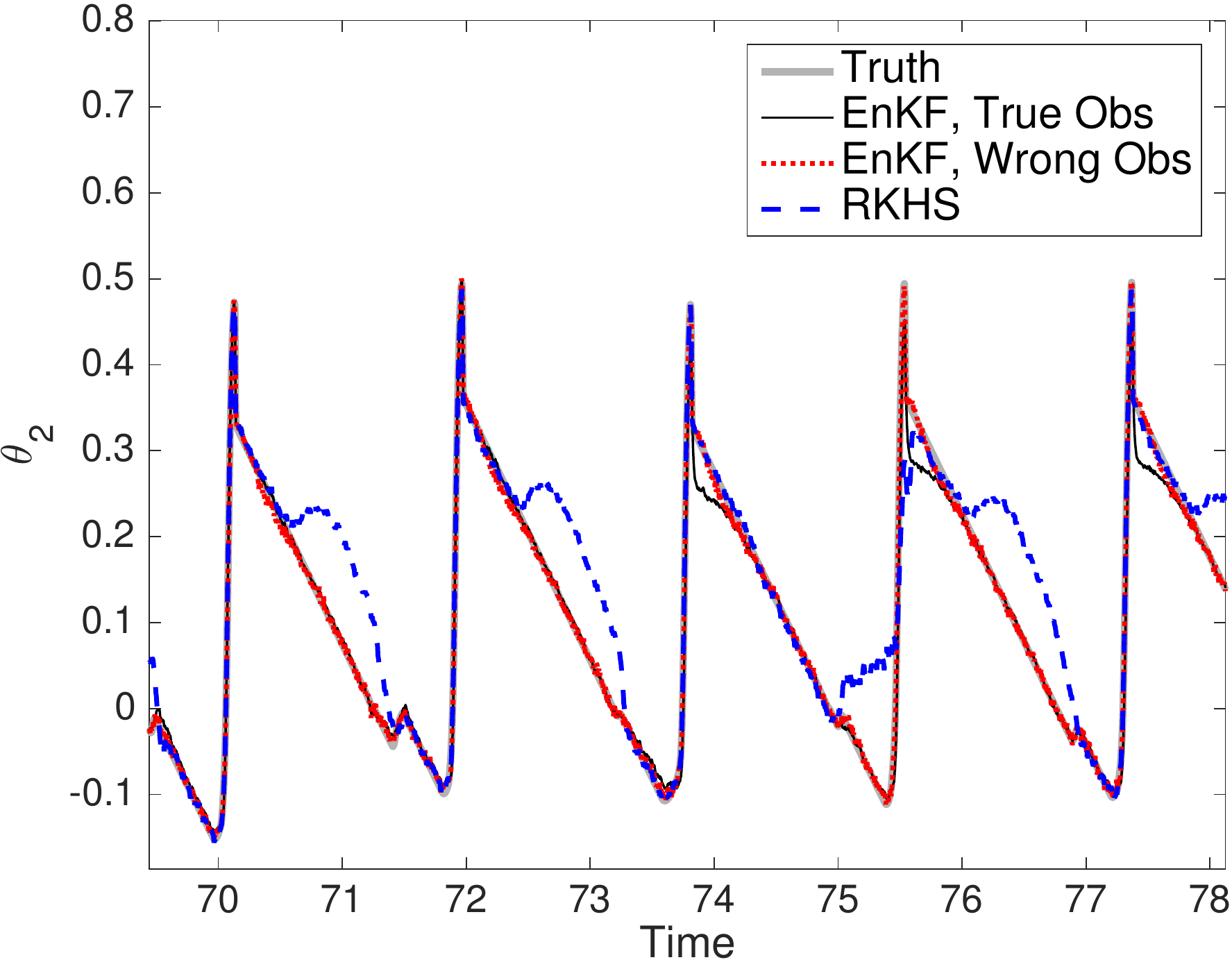}
\includegraphics[width=0.48\linewidth]{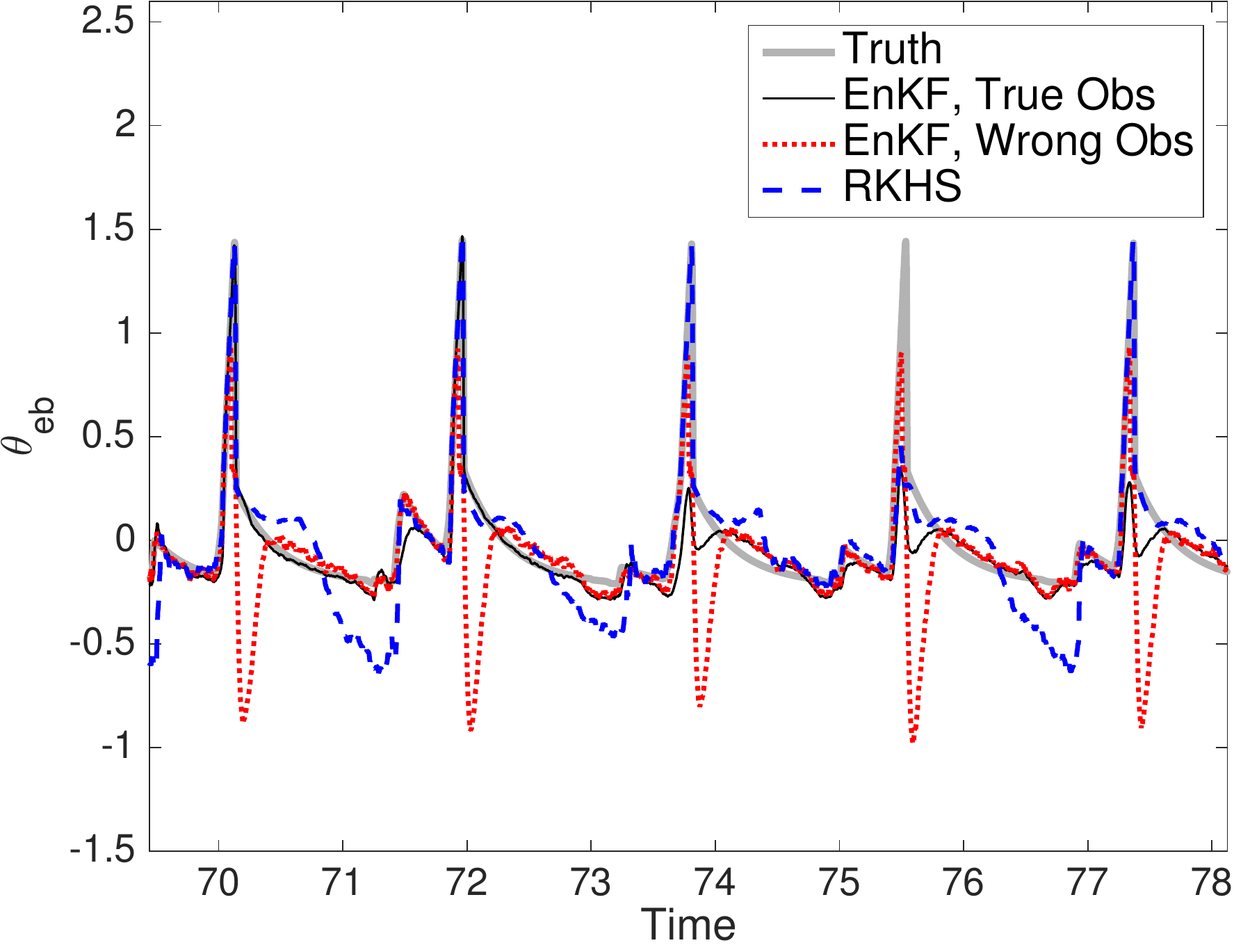}
\includegraphics[width=0.48\linewidth]{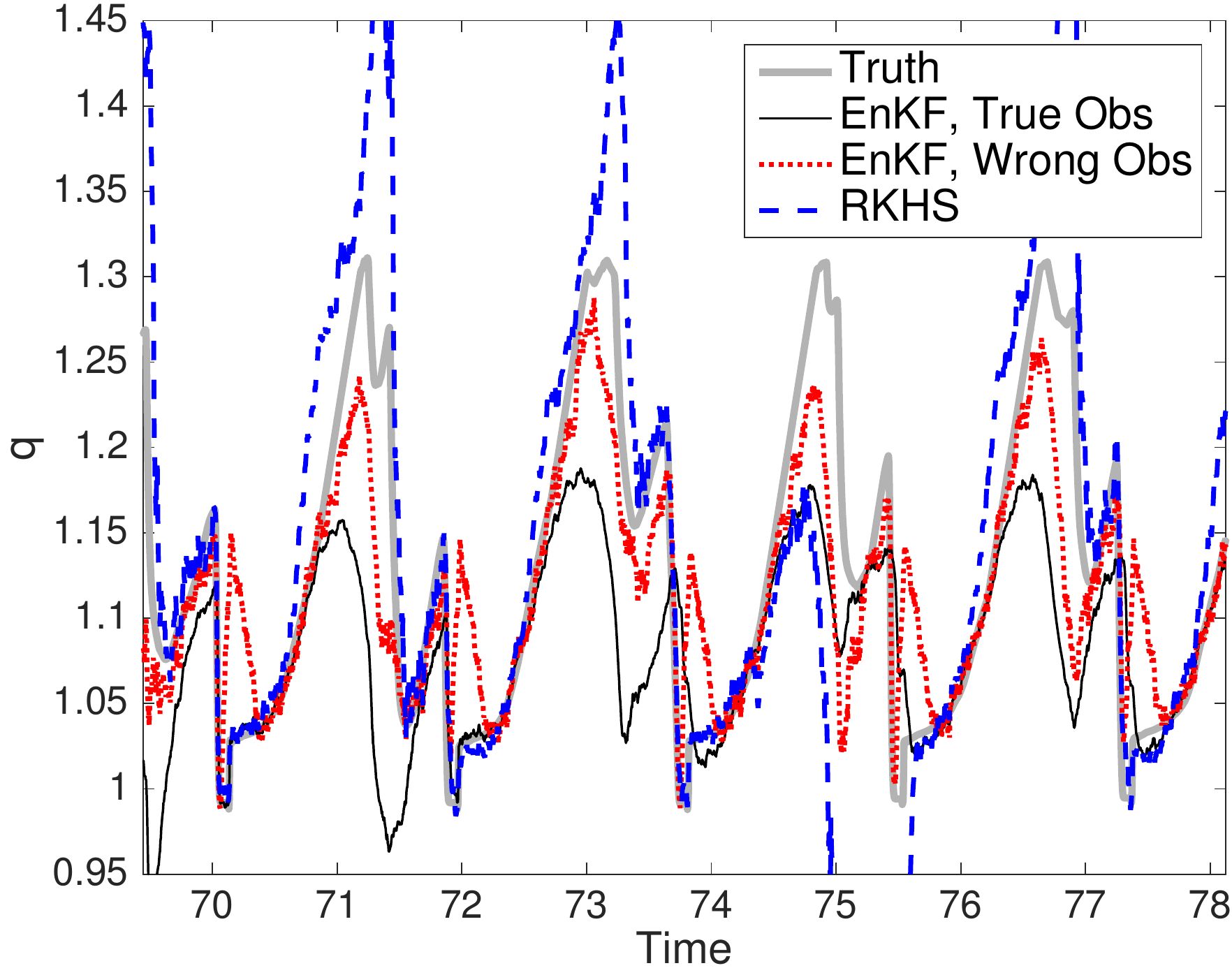}
\includegraphics[width=0.48\linewidth]{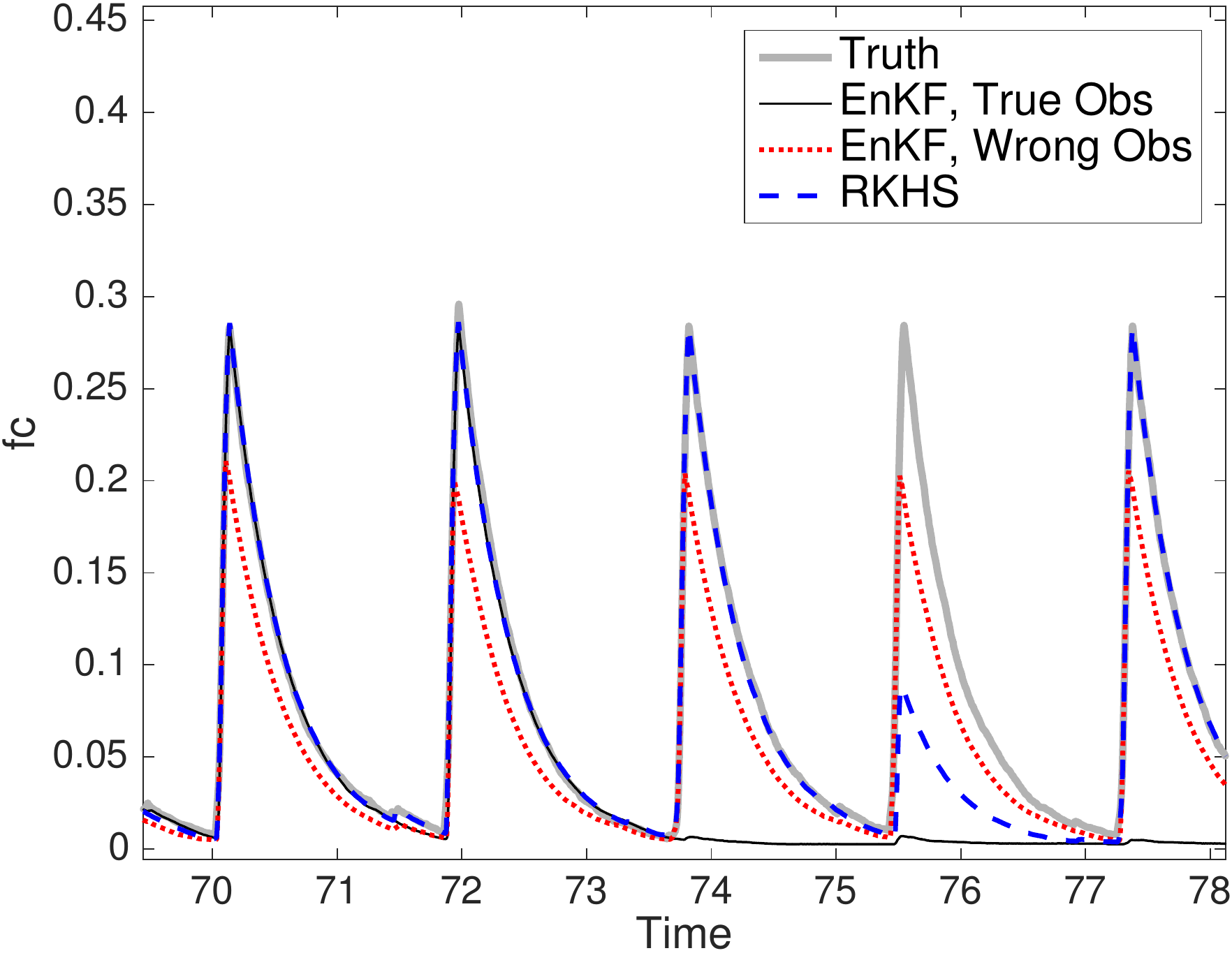}
\includegraphics[width=0.48\linewidth]{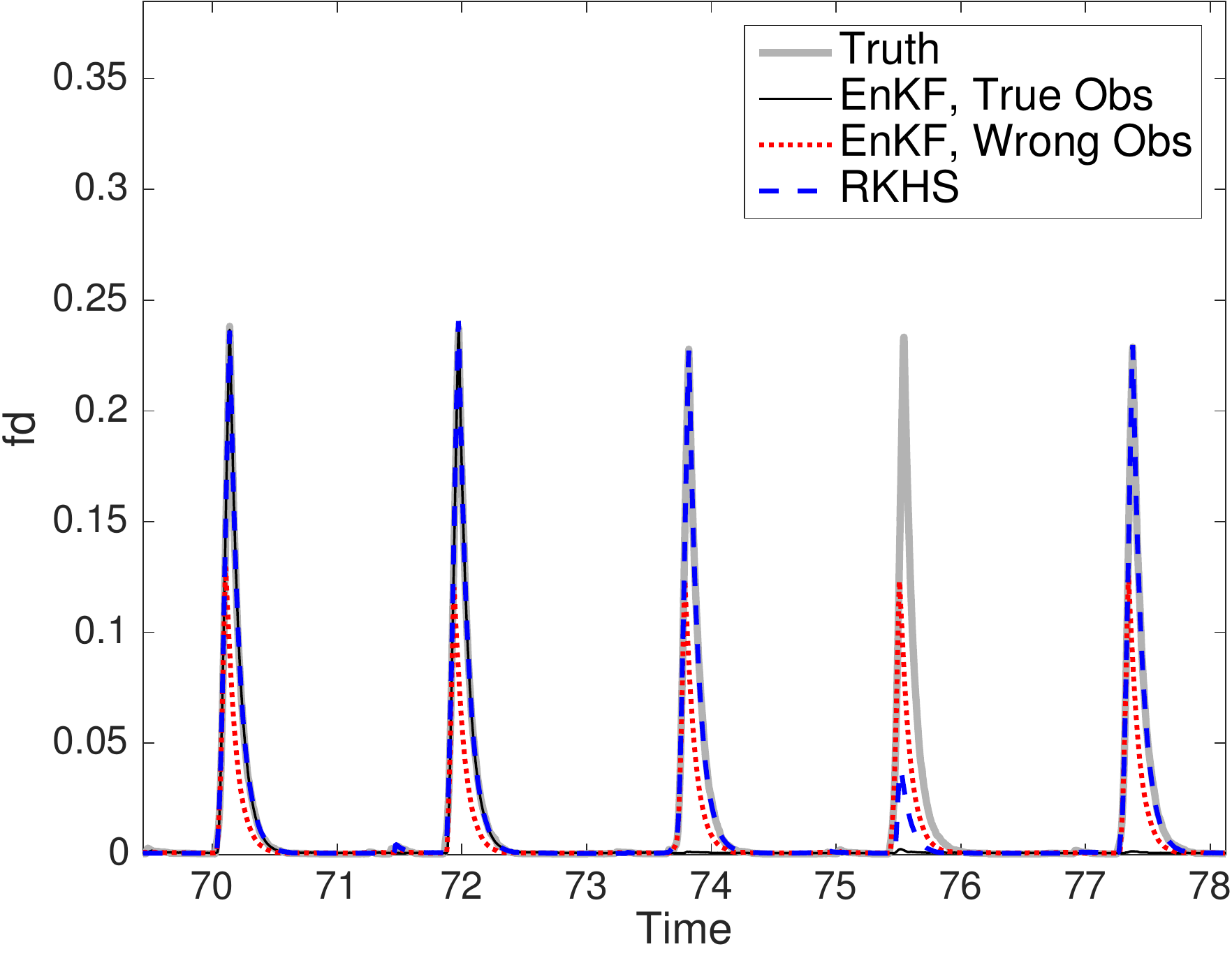}
\includegraphics[width=0.48\linewidth]{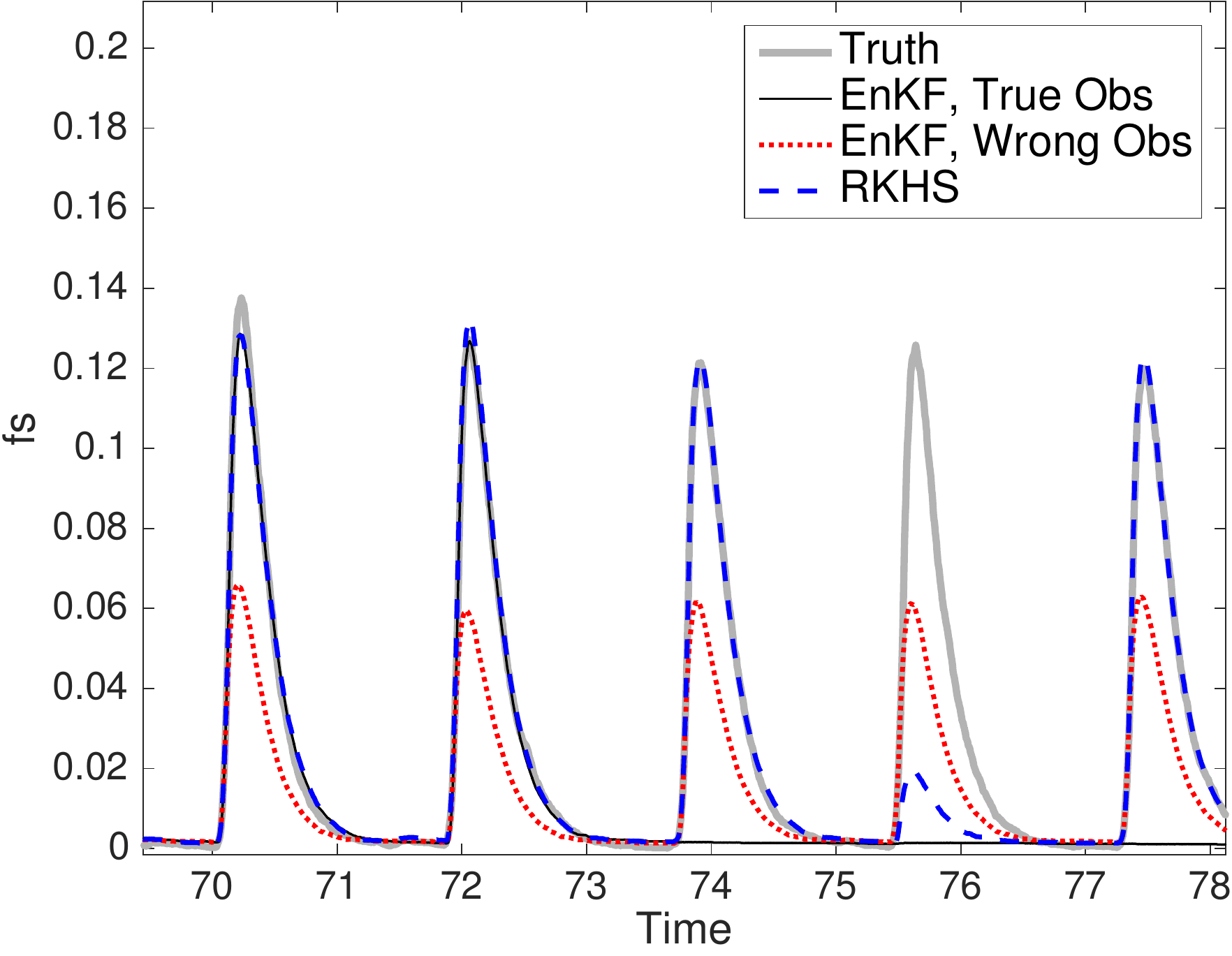}
\caption{\label{cloud2} Time series of filter estimates compared to the true state variables for all 7 model variables in the presence of large measurement error $R^o=0.02 \times \textup{var}(h_\nu(x,f))$ which is 2\% of the observation variance.  Notice that many of the cloud fraction estimates even for the perfect observation model are far from the true state, which indicates the difficulty of filtering these observations when the measurement errors are large.}
\end{figure}

\section{Summary}
We introduced a data-driven observation model error estimator. The estimator is a recursive Bayesian method which uses a nonparametric likelihood function, constructed by representing the kernel embedding of conditional distributions with data-driven basis functions obtained via the diffusion maps algorithm. The method is scalable to high-dimensional problems since it can be used in parallel to correct error in each component of a high dimensional observation. While our main motivation is to solve the cloudy satellite inference problem, the proposed framework can be used in general applications so long as a training data of the states and corresponding observations is available. Also, while the proposed observation model error correction technique is shown in tandem with an ensemble Kalman filter, the fact that we estimate the conditional distribution (and not just the statistics) of the error, $p(b|y)$, implies that this framework can also be used in any filtering method. 

In the examples above we showed that the RKHS correction is able to overcome complex multimodal observation model error.  When there is sufficient observability (eg. short observation time) and small measurement error, the RKHS correction approaches the filtering skill of the EnKF given the perfect observation function.  In the presence of large measurement error or long observation time, the RKHS correction degrades relative to the EnKF with the perfect observation function, however it is still a significant improvement in both skill and stability over simply inflating the EnKF noise parameters, $Q$ and $R$.  The key to the RKHS is using historical data to learn the conditional distribution of the model error.  The RKHS then combines the current filter predicted observation with the actual observation to determine the appropriate correction to the current observation model.  This process depends on choosing an appropriate prior for the current model error, and in particular we presented two natural choices for the covariance of this prior.  Future improvements may be possible by designing better or adaptive priors.


\acknowledgments
The research of J.H. is partially supported by the Office of Naval Research Grants N00014-16-1-2888, MURI N00014-12-1-0912 and the National Science Foundation grant DMS-1317919.


\bibliographystyle{ametsoc2014}
\bibliography{ref}

\begin{thebibliography}{17}
\providecommand{\natexlab}[1]{#1}
\providecommand{\url}[1]{\texttt{#1}}
\renewcommand{\UrlFont}{\rmfamily}
\providecommand{\urlprefix}{URL }
\expandafter\ifx\csname urlstyle\endcsname\relax
  \providecommand{\doi}[1]{doi:\discretionary{}{}{}#1}\else
  \providecommand{\doi}{doi:\discretionary{}{}{}\begingroup
  \urlstyle{rm}\Url}\fi
\providecommand{\eprint}[2][]{\url{#2}}

\bibitem[{Anderson(2001)}]{anderson:01}
Anderson, J., 2001: {An ensemble adjustment Kalman filter for data
  assimilation}. \textit{Monthly Weather Review}, \textbf{129}, 2884--2903.

\bibitem[{Berry and Harlim(2016{\natexlab{a}})Berry, and Harlim}]{bh:16physd}
Berry, T., and J.~Harlim, 2016{\natexlab{a}}: Forecasting turbulent modes with
  nonparametric diffusion models: Learning from noisy data. \textit{Physica D},
  \textbf{320~(57-76)}.

\bibitem[{Berry and Harlim(2016{\natexlab{b}})Berry, and Harlim}]{bh:16vb}
Berry, T., and J.~Harlim, 2016{\natexlab{b}}: Variable bandwidth diffusion
  kernels. \textit{Appl. Comput. Harmon. Anal.}, \textbf{40}, 68--96.

\bibitem[{Berry and Sauer(2013)Berry, and Sauer}]{bs:13}
Berry, T., and T.~Sauer, 2013: {Adaptive ensemble Kalman filtering of nonlinear
  systems}. \textit{Tellus A}, \textbf{65}, 20\,331.

\bibitem[{Bishop et~al.(2001)Bishop, Etherton,, and Majumdar}]{bishop:01}
Bishop, C., B.~Etherton, and S.~Majumdar, 2001: {Adaptive sampling with the
  ensemble transform Kalman filter part I: the theoretical aspects}.
  \textit{Monthly Weather Review}, \textbf{129}, 420--436.

\bibitem[{Coifman and Lafon(2006)Coifman, and Lafon}]{cl:06}
Coifman, R., and S.~Lafon, 2006: Diffusion maps. \textit{Appl. Comput. Harmon.
  Anal.}, \textbf{21}, 5--30.

\bibitem[{Evensen(1994)}]{evensen:94}
Evensen, G., 1994: {Sequential data assimilation with a nonlinear
  quasi-geostrophic model using Monte Carlo methods to forecast error
  statistics}. \textit{Journal of Geophysical Research}, \textbf{99},
  10\,143--10\,162.

\bibitem[{Kalnay(2003)}]{kalnay:03}
Kalnay, E., 2003: \textit{{Atmospheric modeling, data assimilation, and
  predictability}}. Cambridge University Press.

\bibitem[{Khouider et~al.(2010)Khouider, Biello,, and Majda}]{kbm:10}
Khouider, B., J.~Biello, and A.~J. Majda, 2010: A stochastic multicloud model
  for tropical convection. \textit{Commun. Math. Sci.}, \textbf{8~(1)},
  187--216.

\bibitem[{Liou(2002)}]{liou:02}
Liou, K.-N., 2002: \textit{An introduction to atmospheric radiation}, Vol.~84.
  Academic press.

\bibitem[{Lorenz(1996)}]{lorenz:96}
Lorenz, E., 1996: {Predictability - a problem partly solved}.
  \textit{Proceedings on predictability, held at ECMWF on 4-8 September 1995},
  1--18.

\bibitem[{Majda and Harlim(2012)Majda, and Harlim}]{mh:12}
Majda, A.~J., and J.~Harlim, 2012: \textit{Filtering complex turbulent
  systems}. Cambridge University Press.

\bibitem[{McNally(2009)}]{mcnally:09}
McNally, A.~P., 2009: The direct assimilation of cloud-affected satellite
  infrared radiances in the ecmwf 4d-var. \textit{Quarterly Journal of the
  Royal Meteorological Society}, \textbf{135~(642)}, 1214--1229.

\bibitem[{Nystr{\"o}m(1930)}]{nystrom:30}
Nystr{\"o}m, E.~J., 1930: {\"U}ber die praktische aufl{\"o}sung von
  integralgleichungen mit anwendungen auf randwertaufgaben. \textit{Acta
  Mathematica}, \textbf{54~(1)}, 185--204.

\bibitem[{Reale et~al.(2008)Reale, Susskind, Rosenberg, Brin, Liu,
  Riishojgaard, Terry,, and Jusem}]{realeetal:08}
Reale, O., J.~Susskind, R.~Rosenberg, E.~Brin, E.~Liu, L.~P. Riishojgaard,
  J.~Terry, and J.~C. Jusem, 2008: Improving forecast skill by assimilation of
  quality-controlled airs temperature retrievals under partially cloudy
  conditions. \textit{Geophys. Res. Lett.}, \textbf{35~(8)},
  \urlprefix\url{http://dx.doi.org/10.1029/2007GL033002}.

\bibitem[{Song et~al.(2013)Song, Fukumizu,, and Gretton}]{song2013}
Song, L., K.~Fukumizu, and A.~Gretton, 2013: Kernel embeddings of conditional
  distributions: A unified kernel framework for nonparametric inference in
  graphical models. \textit{IEEE Signal Processing Magazine}, \textbf{30~(4)},
  98--111.

\bibitem[{Song et~al.(2009)Song, Huang, Smola,, and Fukumizu}]{song2009}
Song, L., J.~Huang, A.~Smola, and K.~Fukumizu, 2009: Hilbert space embeddings
  of conditional distributions with applications to dynamical systems.
  \textit{Proceedings of the 26th Annual International Conference on Machine
  Learning}, ACM, 961--968.

\end{thebibliography}

\appendix[A]
\appendixtitle{Kernel embedding of distributions}

In this Appendix we discuss the theory of kernel embeddings of distributions for constructing likelihood functions. We should point out that the review in this section follows the derivation in \cite{song2009, song2013} except that we adapt their derivation to weighted Hilbert spaces, which will be required in the next section.

Given a symmetric positive definite kernel, $K:\cal M\times\cal M\rightarrow \mathbb{R}$, the Moore-Aronszajn theorem states that there exists a Reproducing Kernel Hilbert Space (RKHS) ${\cal H}=L^2({\cal M},q)$. This is a space of functions $f:\cal{M}\rightarrow \mathbb{R}$ with a \emph{reproducing property}, that is, for each $x\in\mathcal{M}$, there exists $K(x,\cdot)$ such that $f(x) = \langle f, K(x,\cdot) \rangle_{q}, \forall f\in\cal H$. Moreover, $K(x,\cdot)\in\cal{H}$, which implies that $K(x,y) = \langle K(x,\cdot),K(y,\cdot)\rangle_{q}$.  The map from $\cal M$ to $\cal H$ given by $x \mapsto K(x,\cdot)$ is called the feature map.

Let $X$ be a random variable with distribution $P(X)$ defined on a domain $\cal M$. The \emph{kernel embedding of a distribution $P \in \mathcal{H}$} maps $P$ to a function $\mu_X \in \cal H$ given by,
\BEA
\mu_X := \mathbb{E}_X[K(X,\cdot)] = \int_{\cal M} K(x,\cdot) dP(x)\label{mux}
\EEA
which is the expectation of the feature map.  For any $f\in{\cal H}$,
\BEA
\mathbb{E}_X[f(X)] &=&  \int_{\cal M} f(x)dP(x) = \int_{\cal M} \langle f,K(x,\cdot)\rangle_{q} dP(x)\nonumber \\ &=& \big\langle f, \int_{\cal M} K(x,\cdot) dP(x) \big\rangle_q = \langle f, \mu_X\rangle_q.\nonumber
\EEA
So $\mu_X$ is the Riesz representative of expectation over $P(X)$.

Let $Y$ be another random variable defined on another space $\cal N$ and another kernel $\tilde K(y,\cdot)$, which 
maps $\cal N$ to a feature space $\tilde{\cal H}=L^2({\cal N},\tilde q)$. 
We assume the following equality,
\BEA
\langle K(x,\cdot)\otimes \tilde{K}(y,\cdot),K(x',\cdot)\otimes \tilde{K}(y',\cdot) \rangle_{q \otimes \tilde q} \\ \nonumber = K(x,x')\otimes \tilde K(y,y').\label{ppty}
\EEA
Then the joint distribution $P(X,Y)$ can be mapped into the product feature space ${\cal H} \otimes \cal \tilde H$ with the following definition,
\BEA
\mathcal{C}_{XY} :=  \int_{{\cal M}\times \cal N} K(x,\cdot)\otimes\tilde{K}(y,\cdot)\,dP(x,y).\label{Cxy}
\EEA
Then, for any functions $f\in{\cal H}, g\in{\cal \tilde H}$, we have
\BEA
&& \hspace{-30pt} \mathbb{E}_{XY}[f(X)g(Y)] = \int_{{\cal M}\times \cal N} f(x) g(y) dP(x,y)\nonumber\\ &=& \int_{{\cal M}\times \cal N} \langle f,K(x,\cdot)\rangle_{q}  \langle g,\tilde K(y,\cdot)\rangle_{\tilde q} dP(x,y) \nonumber \\ &=& \int_{{\cal M}\times \cal N}\langle f\otimes g,K(x,\cdot)\otimes \tilde K(y,\cdot)\rangle_{q\times \tilde q}\, dP(x,y)\nonumber\\
&=& \langle f\otimes g, \mathcal{C}_{XY} \rangle_{q\times \tilde q} := \langle f,\mathcal{C}_{XY}g\tilde{q} \rangle_{q},\label{EXYfg}
\EEA
using the equality in \eqref{ppty} and the definition of $\mathcal{C}_{XY}$ in \eqref{Cxy}. Again, the cross-covariance operator $\mathcal{C}_{XY}$ is a Riesz representative of expectation over $P(X,Y)$.

The kernel embedding of conditional distribution $P(Y|X)$ is defined as,
\BEA
\mu_{Y|x} = \mathbb{E}_{Y|x}[\tilde K(Y,\cdot)] =  \int_{\cal N} \tilde K(y,\cdot) dP(y|x),\label{muygivenx}
\EEA
and one can use the same argument to show that the embedding operator $\mu_{Y|x}$ is the Riesz representer of the conditional expectation, 
\BEA
\mathbb{E}_{Y|x}[g(Y)] = \langle g, \mu_{Y|x}\rangle_{\tilde q}.\label{muyx}
\EEA

The main result from \cite{song2013} that we will exploit is that
\begin{theorem} The kernel embedding of $P(Y|X)$ satisfies,
\BEA
\mu_{Y|x} = q\mathcal{C}_{YX}\mathcal{C}_{XX}^{-1}K(x,\cdot),\label{muygivenx2}
\EEA
where the operators $\mathcal{C}_{XY}$ is defined as in \eqref{Cxy} and $\mathcal{C}_{XX} = \int_{\mathcal{M}} K(x,\cdot) K(x,\cdot) dP(x)$. 
\end{theorem}
\begin{proof} This result depends on the identity derived in \cite{song2009}, which states that for $g\in{\cal \tilde H}$,
\BEA
\mathcal{C}_{XX}\mathbb{E}_{Y|X}[g(Y)] = \mathcal{C}_{XY}g\tilde q. \label{fukumizu}
\EEA
To see this, let $f\in {\cal H}$ and notice that
\BEA
\langle f, \mathcal{C}_{XX}\mathbb{E}_{Y|X}[g(Y)] \rangle_{q} &=& \mathbb{E}_{X}[f(X)\mathbb{E}_{Y|X}[g(Y)]] \nonumber \\ &=&\mathbb{E}_{XY}[f(X)g(Y)] = \langle f, \mathcal{C}_{XY}g\tilde q\rangle_{q}, \nonumber
\EEA
where we have used  the fact that $\mathcal{C}_{XX}$ is the Riesz representer of expectation with respect to $P(X)$ and also the equality in \eqref{EXYfg}. For each $\mathbb{E}_{Y|X}[g(Y)]\in{\cal H}$, then by the reproducing property of $\cal H$, one can write,
\BEA
\mathbb{E}_{Y|x}[g(Y)] &=& \langle \mathbb{E}_{Y|X}[g(Y)],K(x,\cdot)  \rangle_{q} 
\nonumber \\ &=& \langle \mathcal{C}_{XX}^{-1}\mathcal{C}_{XY}g\tilde{q}q ,K(x,\cdot)  \rangle\nonumber\\ &=&  \langle g ,q\mathcal{C}_{YX}\mathcal{C}_{XX}^{-1}K(x,\cdot) \rangle_{\tilde q}\label{Eyx}
\EEA
where we used the identity in \eqref{fukumizu} and using the fact that $\mathcal{C}_{XX}$ is symmetric and $\mathcal{C}_{YX} = \mathcal{C}_{XY}^\top$. Comparing \eqref{Eyx} with the \eqref{muyx}, we obtain \eqref{muygivenx2}.
\end{proof}

\appendix[B]
\appendixtitle{Proof of equations~\eqref{muyb}-\eqref{CBB}.}

In this Appendix, we will prove equations~\eqref{muyb}-\eqref{CBB}. For our discussion below, we define two random variables. The first one, $Y\in \mathbb{R}$ with distribution $P(Y)$  and kernel $\tilde{K}$ such that $L^2(\mathbb{R},\tilde{q})$ is the RKHS with orthonormal basis functions $\phi_k(y)$ that can be estimated from the data $y_i$ with sampling density $\tilde{q}(y)$. The second one, $B\in \mathbb{R}$ with distribution $P(N)$  and kernel $K$ such that $L^2(\mathbb{R},q)$ is the RKHS with orthonormal basis functions $\varphi_j(x)$ that can be estimated from the data $b_i$ with sampling density $q(b)$.

Let assume the representation in \eqref{pycondx}. Taking inner-product of \eqref{pycondx} with $\phi_k$ and imposing the orthogonality condition, we obtain,
\BEA
\mu_{Y|b,k} = \langle p(\cdot|b),\phi_k\rangle = \mathbb{E}_{Y|b}[\phi_k].\nonumber
\EEA

Applying \eqref{muyx} and the equality in \eqref{muygivenx} from Theorem~1 with the random variable $X$ replaced by $B$, we have
 \BEA
 \mu_{Y|b,k} &=& \langle \mu_{Y|b}, \phi_k\rangle_{\tilde q} \nonumber \\
 &=&\langle q\mathcal{C}_{YB}\mathcal{C}_{BB}^{-1} K(b,\cdot), \phi_k\rangle_{\tilde q} \nonumber\\
 &=& \langle q\mathcal{C}_{YB}\mathcal{C}_{BB}^{-1}\sum_j  \varphi_j(b)  \varphi_j,  \phi_k\rangle_{\tilde q}.\nonumber \\
 &=& \sum_j  \varphi_j(b) \langle q \mathcal{C}_{YB}\mathcal{C}_{BB}^{-1}  \varphi_j,  \phi_k\rangle_{\tilde q}\nonumber \\ 
 &=& \sum_j  \varphi_j(b)  \langle \mathcal{C}_{YB}\mathcal{C}_{BB}^{-1},  \varphi_j \otimes \phi_k\rangle_{q\otimes \tilde q}. \label{mu_ycondxi}
\EEA
In the second equality above, we used the reproducing property, $\varphi_j(b) = \langle K(b,\cdot), \varphi_j\rangle_q$, such that $K(b,\cdot) = \sum_j  \varphi_j(b)  \varphi_j(\cdot)$. 

Define the projection of the operators $\mathcal{C}_{YB}$ and $\mathcal{C}_{BB}$ on the space spanned by the basis functions $\phi_k(y), \varphi_j(b)$, respectively, with matrices $C_{YB}$ and $C_{BB}$ whose components are,
\BEA
\big[C_{YB}\big]_{jk} &:=& \langle \mathcal{C}_{YB}, \phi_j\otimes\varphi_k\rangle_{\tilde{q}\otimes q} = \mathbb{E}_{YB}[\phi_j \varphi_k],\nonumber \\
\big[C_{BB}\big]_{jk} &:=& \langle \mathcal{C}_{BB}, \varphi_k\otimes\varphi_k\rangle_{q} =  \mathbb{E}_{B}[\varphi_j \varphi_k],\nonumber
\EEA
where the second equality in these two equations follows from \eqref{EXYfg}. Numerically, we can estimate these quantities using Monte-Carlo averages as shown in \eqref{CYB} and \eqref{CBB}, respectively. To complete the proof of equality \eqref{muyb}, we just show that,
\BEA
[C_{YB}C_{BB}^{-1}]_{kj} &:=& \sum_l [C_{YB}]_{kl} [C_{BB}^{-1}]_{lj} \nonumber \\&=& \sum_l \langle \mathcal{C}_{YB},\phi_k\otimes\varphi_l \rangle_{\tilde q\otimes q}\langle \mathcal{C}_{BB}^{-1},\varphi_l \varphi_j \rangle_{q} \nonumber \\
&=&   \left\langle \mathcal{C}_{YB},\phi_k\otimes \left( \sum_l \langle \mathcal{C}_{BB}^{-1},\varphi_l \varphi_j \rangle_{q}\varphi_l \right) \right\rangle_{\tilde q\otimes q} \nonumber \\
&=&   \left\langle \mathcal{C}_{YB},\phi_k\otimes \mathcal{C}_{BB}^{-1}\varphi_j \right\rangle_{\tilde q\otimes q} \nonumber \\
&=&   \left\langle \mathcal{C}_{YB}\mathcal{C}_{BB}^{-1},\phi_k\otimes \varphi_j \right\rangle_{\tilde q\otimes q}\label{matmul}.
\EEA
Together with \eqref{mu_ycondxi} and the Monte-Carlo approximations on $C_{YB}$ and $C_{BB}$, we obtain \eqref{muyb}-\eqref{CBB}. In the derivation above, we used infinite number of basis functions. Numerically, we use finite number of basis functions as shown in the main text.

\appendix[C] 
\appendixtitle{A simplified radiative transfer model}

In this Appendix, we discuss the simplified radiative transfer model used in Section~5. Given the physical variables in the stochastic cloud model in \cite{kbm:10}, the first and second baroclinic potential temperatures, $\theta_1$ and $\theta_2$, respectively, the equivalent boundary layer potential temperature, $\theta_{eb}$, and the vertically average water vapor content, $q$, we consider an idealistic radiative transfer model for the brightness temperature at wavenumber $\nu$ at clear-sky condition, following the standard formulation in \cite{liou:02}:
\BEA
h_\nu(x) = \theta_{eb} T_\nu(0) + \int_0^\infty T(z) \frac{\partial T_\nu}{\partial z} (z)\,dz, \label{TB}
\EEA
where $x=\{\theta_1,\theta_2,\theta_{eb},q\}$ and the temperature at height $z$ is defined by as in \eqref{temperature}. The first term on the right hand side of \eqref{TB} denotes the contribution from the surface and the second term denotes the contribution from the whole column of the atmosphere which is a weighted average of the temperature with weighting function,  $\frac{\partial T_\nu}{\partial z}$.

In \eqref{TB}, we define the transmission between heights $z$ to $\infty$, assuming the height of the satellite instrument is much larger than the troposphere, as follows,
\BEA
T_\nu(z) &=& \exp(- \int_z^\infty \alpha_\nu(s)\,ds).\nonumber
\EEA
Here $\alpha_\nu$ denotes the absorption rate which we assumed to decrease exponentially as a function of height,
\BEA
\alpha_\nu(z) = \alpha_\nu^*q\exp(-\frac{z}{H}),\nonumber
\EEA
where $\alpha_\nu^*$ denotes a reference absorption rate that is to be determined.  In our simulation, we set $H=3$ km.
With this assumption, we have,
\BEA
T_\nu(z) &=& \exp(- \alpha_\nu(z)H)\nonumber
\EEA
and $T_\nu(0) = \exp(- \alpha_\nu^*qH)$.
Also, the weighting functions become,
\BEA
\frac{\partial T_\nu}{\partial z} (z) &=& \alpha_\nu(z) T_\nu(z),\nonumber
\EEA

We determine the wavenumber $\nu$ by specifying $\alpha_\nu^*$ corresponding to the height of which the weighting function is maximum. That is, setting $\frac{\partial T_\nu^2}{\partial z^2} (z) = 0$, we have
\BEA
\alpha_\nu^* = \frac{\exp (\frac{z_{max}}{H}) }{q H},\nonumber
\EEA
where $q$ needs to be fixed to a specific reference value. To increase the sensitivity of the weighting function to $q \in[a,b]$ where $a,b$ denote the lower and upper bounds of $q$ which can be obtained by the climatological data, we rescale $q$ to fluctuate in between $[1,2]$ by defining
\BEA
\tilde{q} = \frac{q-a}{b-a}+1.\nonumber
\EEA
and define the absorption rate to be,
\BEA
\alpha_\nu(z) = \alpha_\nu^*\tilde{q}\exp(-\frac{z}{H}),\quad\quad
\alpha_\nu^* = \frac{\exp (\frac{z_{max}}{H}) }{H}.\nonumber
\EEA
This means that when $q = a$, the corresponding $\alpha_\nu(0)$ will produce a weighting function, $\frac{\partial T_\nu}{\partial z} (z)$ that has a maximum weight at height $z_{max}$. In Fig.~\ref{weight}, we show the weighting functions (black solid) with the reference humidity $q=a$ associated with $z_{max}=2, 5, 8, 10$ km (black dashes). We also include the weighting functions corresponding to humidity value $q=b$ (red solid), which show that depending on the value of $q$, the shape of weighting functions will vary between these two weights.

For a single type of cloud, the basic equation for the RTM of the cloudy sky measurement with cloud fraction $c$ and cloud top height $z_t$ is given as follows,
\BEA
h_\nu(x,c) &=& (1-c)\Big( \theta_{eb} T_\nu(0) + \int_0^{z_t} T(z) \frac{\partial T_\nu}{\partial z} (z)\Big)\,dz  \nonumber \\ &&+ cT(z_t) T_\nu(z_t) + \int_{z_t}^\infty T(z) \frac{\partial T_\nu}{\partial z} (z)\,dz\label{TB1cloud}
\EEA
So, the first term denotes the contribution below the cloud, the second term denotes the contribution from the cloud top height, and the last term denotes the contribution from the clear sky above the cloud top height.

The stochastic component of the cloud model in \cite{kbm:10} is a birth-death process that accounts for the evolution of the cloud fractions, $f=\{f_c, f_d, f_s\}$ of three cloud types, cumulus, deep, and stratiform clouds, respectively. Assuming that the deep and stratiform cloud top heights are similar, namely $z_d=12$km, and the cumulus cloud top height is $z_c=3$km, we consider a two-cloud type formulation as follows, 
\BEA
h_\nu(x,f) &=& (1-f_d-f_s)\Big[ \theta_{eb} T_\nu(0) + \int_0^{z_d} T(z) \frac{\partial T_\nu}{\partial z} (z)\,dz\Big]  \nonumber \\ & &+ (f_d+f_s)T(z_d) T_\nu(z_d) + \int_{z_d}^\infty T(z) \frac{\partial T_\nu}{\partial z} (z)\,dz \nonumber\\
&=& (1-f_d-f_s)\Big[(1-f_c)\big( \theta_{eb} T_\nu(0) \nonumber \\ &&+ \int_0^{z_c} T(z) \frac{\partial T_\nu}{\partial z} (z)\,dz \big) \nonumber\\&&+ f_c T(z_c) T_\nu(z_c) + \int_{z_c}^{z_d} T(z) \frac{\partial T_\nu}{\partial z} (z)\,dz \Big]  \nonumber \\ & &+ (f_d+f_s)T(z_d) T_\nu(z_d) + \int_{z_d}^\infty T(z) \frac{\partial T_\nu}{\partial z} (z)\,dz\nonumber
\EEA
Here, the first two rows denotes the contribution below the deep and stratiform clouds, the first term in the last row denotes the contribution from the deep cloud top height, and the last term denotes the contribution from the atmospheric column above the cloud. One can check that if $f_c=0$ and/or $f_d+f_s=0$, then this formulation is consistent with \eqref{TB} and \eqref{TB1cloud}.

In our numerical simulations, we simulate the observations as follows (ignoring time index-$i)$,
\BEA
y_\nu = h_\nu(x,f) + \eta_\nu, \quad \eta_\nu\sim\mathcal{N}(0,R^o),\nonumber
\EEA
for 16 wavenumbers $\nu$ chosen such that the weighting functions, $\frac{\partial T_\nu}{\partial z}$, are maximum at heights $1, 2, \ldots, 16$ km. But assuming that we don't know the cloud top heights, $z_c, z_d$, as well as the cloud fractions, $f=\{f_c,f_d,f_s\}$, we will apply the filtering with the clear-sky observation model, $\tilde{h}_\nu(x) = h_\nu(x,0)$, 
\BEA
y_\nu \approx \tilde{h}_\nu(x) + b_\nu + \eta_\nu, \quad \eta_\nu\sim\mathcal{N}(0,R),\nonumber
\EEA
where $b_\nu$ denotes error for at frequency-$\nu$. In our implementation, we will either fix the parameter $R$ or estimate it adaptively using the method from \cite{bs:13}.



\end{document}